\documentclass[12pt]{article}
\usepackage[utf8]{inputenc}

\usepackage{svgcolor}
\def\showauthornotes{1}
\ifnum\showauthornotes=1

\newcommand{\jnote}[1]{\textcolor{blue}{ {\textbf{(Jeff: #1)}}}}
\newcommand{\snote}[1]{\textcolor{cyan}{ {\textbf{(SiuOn: #1)}}}}
\newcommand{\tnote}[1]{\textcolor{red}{ {\textbf{(Tommaso: #1)}}}}
\else

\newcommand{\jnote}[1]{}
\newcommand{\snote}[1]{}
\newcommand{\tnote}[1]{}
\fi

\def \R {\mathbb{R}}
\def \C {\mathbb{C}}

\def \E {\mathop{\mathbb{E}}}

\def \S {\mathbb{S}}
\def \eps {\epsilon}

\renewcommand{\Pr}{\mathop{\bf Pr\/}}

\newcommand{\ot}{\otimes}

\newcommand{\Var}{\mathrm{Var}}

\newcommand{\pE}{\mathop{\widetilde{\E}}}

\usepackage{amsmath,amssymb,amsthm,amsfonts,latexsym,bm,xspace,graphicx,float,mathtools,mathdots,physics}
\usepackage{braket,caption,subcaption,ellipsis,xcolor,textcomp,hhline,pifont,combelow,booktabs}
\usepackage[colorlinks=true, allcolors=blue]{hyperref}
\usepackage{color}
\usepackage{times}
\usepackage{fullpage}
\usepackage{tikz-cd}
\usepackage[shortlabels]{enumitem}
\usepackage{thm-restate}
\usepackage[capitalise]{cleveref}
\usepackage{mdframed}
\usepackage{csquotes}
\usepackage{cancel}

\newtheorem{theorem}{Theorem}[section]
\newtheorem{lemma}[theorem]{Lemma}
\newtheorem{claim}[theorem]{Claim}
\newtheorem{proposition}[theorem]{Proposition}

\newtheorem{corollary}[theorem]{Corollary}

\newtheorem{definition}[theorem]{Definition}
\newtheorem{remark}[theorem]{Remark}

\newcommand{\poly}{\operatorname{poly}}

\newcommand{\supp}{\textnormal{supp}}

\DeclareMathOperator{\Span}{span}
\renewcommand{\Re}{\mathop{\textnormal{Re}}}

\DeclareMathAlphabet{\mathsfit}{\encodingdefault}{\sfdefault}{m}{sl}
\SetMathAlphabet{\mathsfit}{bold}{\encodingdefault}{\sfdefault}{bx}{n}

\newcommand\sC{\mathcal{C}}

\newcommand{\defas}{\coloneqq}

\newcommand{\Paren}[1]{\left(#1\right)}
\newcommand\Brac[1]{\left[#1\right]}
\newcommand\Abs[1]{\left\lvert#1\right\rvert}

\newcommand{\Norm}[1]{\left\lVert#1\right\rVert}

\newcommand{\inner}[1]{\langle#1\rangle}
\usepackage[only,llbracket,rrbracket]{stmaryrd}
\newcommand{\Iv}[1]{\left\llbracket#1\right\rrbracket}

\newcommand{\cB}{\mathcal B}
\newcommand{\cC}{\mathcal C}

\newcommand{\cE}{\mathcal E}

\newcommand{\cH}{\mathcal H}
\newcommand{\cI}{\mathcal I}

\newcommand{\cL}{\mathcal L}

\newcommand{\cN}{\mathcal N}

\newcommand{\cT}{\mathcal T}

\newcommand{\cW}{\mathcal W}

\mathchardef\mhyphen="2D
\newcommand{\val}{\text{val}}

\renewcommand{\int }{\text{Int}}

\newcommand{\Erdos}{Erd\H{o}s\xspace}
\newcommand{\Renyi}{R\'enyi\xspace}
\newcommand{\Lovasz}{Lov\'asz\xspace}

\usepackage[notes=true,later=false,camera=false]{dtrt}
\boolfalse{dt@markundef}
\usepackage{algorithm}
\usepackage{algorithmic}

\allowdisplaybreaks

\begin{document}
\title{Strong Refutation for Random Quantum 3-SAT at Constant Density}
\author{
  Siu On Chan
  \thanks{sochan@gmail.com}
  \and
  Jeff Xu
  \thanks{{Toyota Technological Institute at Chicago}.{jeffxusichao@ttic.edu}}
}

\maketitle
\thispagestyle{empty}

\abstract{
  We give a classical polynomial-time algorithm that strongly refutes random quantum $3$-SAT at sufficiently large
  constant constraint density, thereby disproving the quantum analogue of Feige's random $3$-SAT hypothesis.
  This stands in sharp contrast to classical random $3$-SAT, for which polynomial-time strong refutation is known only
  at constraint density $\Delta\gtrsim n^{1/2}$.
  Although quantum $3$-SAT shares the pairwise-independence barrier of its classical counterpart, our SDP-based
  refutation overcomes this barrier by exploiting the noncommutativity of quantum constraints.
}

\vfill
\paragraph{AI Disclosure:} The authors formulated the two key ingredients of the refutation algorithm: the reduction to pairwise-independent
states through the classical analogue, and the strong refutation of pairwise-uniform states using bow-tie structures.
AI assistance was used to work out the technical details, with the specific contributions highlighted in a later section. The authors verified the correctness and originality of all content, including the references.

\clearpage\newpage
\thispagestyle{empty}
\setcounter{page}{0}

\tableofcontents
\clearpage\newpage
\thispagestyle{empty}
\setcounter{page}{1}

\clearpage
\newpage

\section{Introduction}

Quantum \(k\)-SAT, or \(k\)-QSAT, is the quantum analogue of classical \(k\)-SAT, introduced by Bravyi~\cite{bravyi2011efficient}. For any \(k\) variables, a classical \(k\)-SAT clause selects one of the \(2^k\) sign patterns and thereby forbids a single assignment in \(\{0,1\}^k\), equivalently, one of the \(2^k\) computational-basis directions. A \(k\)-QSAT clause is the continuous analogue: instead of forbidding a basis direction, it forbids an arbitrary one-dimensional subspace of the local Hilbert space \((\C^2)^{\otimes k}\); in the random model, this forbidden direction is sampled from the Haar measure. Thus, for \(k=3\), the eight discrete forbidden directions of classical \(3\)-SAT are replaced by a continuum of directions in an eight-dimensional Hilbert space.

In the same work, Bravyi gave a classical polynomial-time algorithm for \(2\)-QSAT and proved that \(k\)-QSAT is \(\mathrm{QMA}_1\)-complete for \(k\ge 4\).
Gosset and Nagaj later closed the remaining gap by proving \(\mathrm{QMA}_1\)-completeness of \(3\)-QSAT~\cite{gosset2016quantum}.
The complexity class $\text{QMA}_1$ was introduced in \cite{bravyi2011efficient} as a quantum analogue of $\text{NP}$,
consisting of problems in $\text{QMA}$ whose Yes instances can be accepted by a quantum verifier with probability $1$.
$\text{QMA}_1$ was recently shown to be the same as $\text{QMA}$, concurrently by Grewal and Rudolph
\cite{grewal2026qma} and by Liu and Vidick \cite{liu2026achieving}.
So the complexity of $k$-QSAT seems to nicely mirror that of $k$-SAT, when the complexity classes are appropriately
quantized.

Much like the classical countperpart (e.g.~\cite{ding2015proof}), the satisfiability threshold of random $k$-QSAT has
also been studied.
Ravelomanana \cite{ravelomanana2026slice} upper bounded the random $3$-QSAT threshold by $1.947$ and the $k$-QSAT
threshold by $O(2^k/k)$, improving on previous works of Laumann, L\"auchli, Moessner, Scardicchio, and Sondhi
\cite{laumann2010product,laumann2010random} and Bravyi, Moore, and Russell \cite{bravyi2010bounds}.
On the other hand, Ambainis, Kempe, and Sattath \cite{ambainis2012quantum} lower bounded the random $k$-QSAT threshold
by $\Omega(2^k/k^2)$ using the quantum \Lovasz local lemma.

Several recent works study algorithmic variants of Hamiltonian satisfiability and optimization.  Montanaro, Shao, and
Verdon~\cite{montanaro2025testing} gave a randomized polynomial-time property tester for quantum $k$-SAT.  Hastings and
O'Donnell~\cite{hastings2022optimizing} obtained efficient classical constant-factor upper certificates for the dense
degree-$4$ fermionic Sachdev--Ye--Kitaev model, while Kocurek~\cite{kocurek2026spectral} gave a classical spectral
algorithm for strongly refuting semirandom Hamiltonian $k$-XOR.  Chen, Dalzell, Berta, Brand\~ao, and
Tropp~\cite{chen2024sparse} showed that quantum phase estimation efficiently prepares near-ground-energy states of
certain sparse random but generally nonlocal Hamiltonians.  These and many other results reveal the diverse algorithmic
phenomena exhibited by random and structured Hamiltonian models.

In this work, we study the average-case task of strong refutation, motivated by its classical counterpart. A central conjecture in average-case complexity, Feige's hypothesis~\cite[Hypothesis~2]{Fei02}, posits the hardness of strongly refuting random classical $3$-SAT:
\begin{displayquote}
For every $\eps>0$, and for sufficiently large $\Delta$, there is no polynomial-time algorithm that certifies that most random $3$-SAT instances on $n$ variables and $\Delta n$ constraints have value at most $1-\eps$.
\end{displayquote}

This naturally suggests a quantum analogue.
Following \cite{laumann2010product,laumann2010random},
we consider random quantum $k$-SAT instances generated in direct analogy with classical random $k$-SAT:
\begin{enumerate}
    \item The underlying hypergraph is sampled exactly as in the binomial/\Erdos--\Renyi model for classical random $k$-SAT.

    \item In classical $k$-SAT, each clause forbids a uniformly random computational-basis direction, equivalently one of the $2^k$ assignments in $\{0,1\}^k$. In quantum $k$-SAT, this discrete random choice is replaced by a Haar-random one-dimensional subspace of the local Hilbert space $(\C^2)^{\otimes k}$.
\end{enumerate}

Thus, the random hypergraph structure is unchanged, while the discrete randomness in the clause signings is replaced by continuous Haar randomness. In the worst-case setting, hardness of quantum \(k\)-SAT follows immediately from that of classical \(k\)-SAT, since classical clauses correspond to computational-basis directions. The average-case setting, however, is substantially less clear: under the Haar measure, the probability that a randomly chosen forbidden direction coincides with any computational-basis direction is $0$.

Perhaps surprisingly, and in sharp contrast with the intuition suggested by classical CSPs, we show that the natural quantum analogue of Feige's hypothesis fails dramatically:

\begin{restatable}{theorem}{MainTheorem}
  \label{thm:main}
  There is $\delta > 0$ and a classical deterministic polynomial-time algorithm that, for every sufficiently large constant
  $\Delta$, certifies that a random $3$-QSAT instance on $n$ variables with $\Delta n$ constraints in expectation has
  value at most $1-\delta$, with high probability over the instance.
\end{restatable}
%

This behavior stands in stark contrast with classical random $3$-SAT. Currently, polynomial-time strong refutation of random $3$-SAT requires
\(
\Delta \gtrsim n^{1/2}
\)
~\cite{AOW15, dorsi_trevisan23}. The relevant classical refutation algorithms are based on spectral methods or, more generally, Sum-of-Squares hierarchy. The emergence of the $n^{1/2}$ threshold has a structural explanation: every $3$-SAT clause supports a pairwise independent distribution over its satisfying assignments. More generally, for classical random $k$-CSPs, the spectral/SDP refutation threshold is governed by the amount of local independence supported by the predicate. If $t$ is the largest integer for which the satisfying assignments support a $t$-wise independent distribution, then polynomial-time spectral or SDP refutation requires clause density
\(
\Delta \geq \widetilde{\Omega}\!\left(n^{(t-1)/2}\right),
\)
and this threshold is tight~\cite{BCK15, KMOW17, ChanDorsiXu26}.

Our theorem therefore reveals a qualitative separation between the classical and quantum settings: although, as in classical $3$-SAT, each random quantum constraint supports a pairwise-uniform solution in the appropriate sense (see \cref{prop:pairwise-uniform-avoiding-vector}), pairwise independence \cite{benabbas12sdp, tulsiani2013lsplus, BCK15, KMOW17} is no longer a barrier, and
Haar-random quantum constraints can nevertheless be strongly refuted at constant clause density. 
%
%

\subsection{Proof Overview} \label{sec:overview}

Due to the superficial similarities with classical $3$-SAT, it is tempting to try to refute random quantum $3$-SAT by
the approach of Allen, O'Donnell, and Witmer \cite{AOW15}.
However, this approach does not work for $3$-QSAT for two reasons.
\cite{AOW15} considered the following induced distribution of a classical $3$-SAT instance $\cI = (V, \cC)$ and
assignment $x: V \to \set{0,1}$:
\( \E_{c \in \cC} (x_S \oplus b) \,, \)
where $S \in V^3$ is the scope of the constraint $c$ and $b \in \set{0,1}^3$ are the literal negations of $c$.
In words, their induced distribution is the average, over the constraints $c$ of $\cI$, of the local assignment $x_S$ to
$c$, shifted by the negations $b$.
\cite{AOW15} proved that the induced distribution is close to $3$-wise uniform whp when
$\Delta \geq \tilde\Omega(n^{1/2})$, by bounding its Fourier coefficients.

By contrast, quantum $3$-SAT does not have negations, so the above induced distribution does not apply.
One may naively think that the forbidden pure state $\psi$ in a $3$-QSAT plays the same role as the negations $b$ in
$3$-SAT, but in fact the literals $b$ contains more information than $\psi$: Given any $3$-SAT constraint with negations
$b$, one can turn it into the unnegated constraint (that forbids the assignment $000$) in a unique way by adding $b$
entrywise, as in the above definition of induced distribution.
On the other hand, given any $3$-QSAT constriant that forbids the pure state $\psi$, there are multiple ways to
rotate this constraint into the standard constraint (that forbids $\ket{000}$).
The non-uniqueness hinders generalizing the above induced distribution definition to $3$-QSAT.
Another issue is that we want a strong refutation at $\Delta = O(1)$, an unrechable regime to the classical 3-SAT.

\paragraph{Reduction to Pairwise Independent States}
For this reason, we instead consider the induced distribution from \cite{ChanDorsiXu26} and its quantum analog.
Given $3$-QSAT instance $\cI = (V, \cE)$ and an $n$-qubit state $\rho$, the induced state
\( \cE(\rho) = \E_{e \in \cE} \rho_e \)
is the average, over the constraints $e$, of the reduced state $\rho_e$ on $e$.
This induced state is well defined for any degree-$3$ pseudo-state $\rho$, not just a genuine state.

In \cref{sec:concentration}, we quantize the concentration arguments of \cite{ChanDorsiXu26}, and prove that in our
setting, the pairwise marginals $\cE(\rho)_{ij}$ are close to \( \E_{(i,j) \in {V \choose 2}} \rho_{ij}\,. \)
Then we use a quantum variant of the global correlation rounding argument \cite{BRS11} to deduce that the conditional
pairwise marginals $\cE(\rho|_b)_{ij}$ are close to the product state $\cE(\rho|_b)_i \otimes \cE(\rho|_b)_j \approx
\nu(\rho|_b)^{\otimes 2}$, upon appropriately chosen conditions $b$.
A similar quantum variant appeared in \cite[Theorem~14]{BH13}, but we cannot use it as is, for its lack of the law of
total probability $\E_b \cE(\rho|_b) = \cE(\rho)$, due to its use of classical conditioning \cite[Step 4 of
Algorithm~24]{BH13}.
Our \cref{thm:cond-indep-complete} instead works directly with quantum conditioning, and recovers an
approximate law of total probability.

The previous paragraph shows that $\cE(\rho|_b)$ is almost pairwise $\nu$-independent.
Our next step shows that $\cE(\rho|_b)$ is close to some state pairwise $\nu$-independent state $\tau$.
Similar statements have been shown for classical distributions in \cite{ChanDorsiXu26,ACLM26}, generalizing previous
results for the special case when $\nu$ is the uniform distribution on the boolean domain
\cite{alon03almost,odonnell18closness} or is somewhat unbiased on a general finite domain \cite{rubinfeld13robust}.
We prove the first quantum analog of this result in \cref{thm:almost-twise-indep}.
Working with quantum states is more challenging than with classical distributions, when non-negativity constraints are
replaced with positive-semidefiniteness constraints, and due to non-commutivity and quantum coherence.
Because of the usefulness of this result in the classical setting, our quantum version is of independent interest.
%

\paragraph{From Pairwise Independence to Pairwise Uniformity}
It remains to understand which one-qubit marginals $\nu$ can support a high-value pairwise $\nu$-independent state. For classical CSPs with uniformly random signings, symmetry makes the answer immediate: it suffices to consider the uniform marginal. In the absence of random signings, however, there is no a priori reason for $\nu$ to be uniform. A key observation in our setting is that the rotational symmetry of the Haar measure plays an analogous role. We show that any non-uniform marginal necessarily leaves a constant amount of room for strong refutation. Consequently, it suffices to consider SDP solutions whose local marginals are close to uniform, reducing pairwise independence to pairwise uniformity. This is formalized in \cref{lem:biased-marginal}, which shows that every sufficiently biased marginal admits a strong refutation.

\paragraph{Refutation for Pairwise-Uniform States}
Finally, by the standard global-correlation argument, we may reduce to the case where any high-value SDP solution is close to pairwise uniform. This is well aligned with the intuition from classical CSPs: pairwise-uniform solutions are precisely the genuine obstruction to strong refutation, and indeed underlie lower bounds that persist all the way up to clause density $\Delta \lesssim \sqrt{n}$. Thus, in retrospect, the main challenge is naturally expected to lie in refuting pairwise-uniform states.

\begin{figure}[H]
  \centering
  \includegraphics[width=0.58\textwidth]{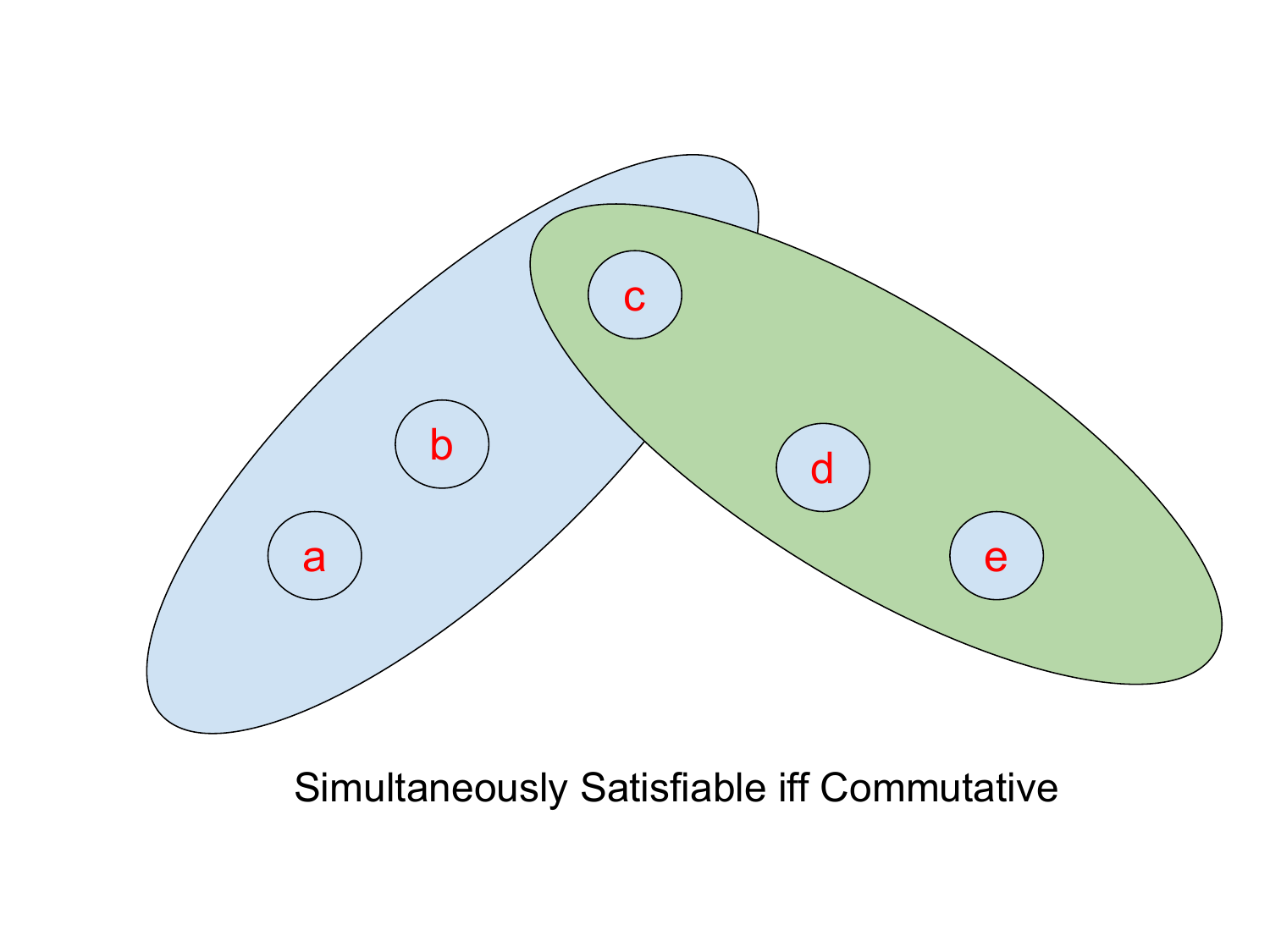}
  \caption{A bow-tie formed by two $3$-local clauses on $\{a,b,c\}$ and $\{c,d,e\}$.}
  \label{fig:bowtie}
\end{figure}

To overcome this obstacle, we exploit a simple local structure formed by two adjacent clauses, which we call a \emph{bow-tie}: a pair of $3$-QSAT constraints whose supports overlap in exactly one qubit; see \cref{fig:bowtie}. We begin with the special case in which the forbidden directions are rank one. Namely, suppose
\(
    \psi_e = u_a\otimes v_b\otimes w_c \in \C^8,
\)
for $u_a,v_b,w_c\in\C^2$, and similarly for the neighboring clause $\psi_f= w'_c\otimes v_d \otimes u_e$ . This may be viewed as a mild generalization of a classical clause: the local factors are no longer restricted to computational-basis vectors. Already in this restricted setting, however, the behavior is qualitatively different from its classical analogue.

In particular, we show that unless $w_c$ and $w_c'$ are co-linear, for two constraints sharing exactly one qubit, with rank-one forbidden directions $\psi_e$ and $\psi_f$, any state $\rho$ that is sufficiently close to pairwise uniform must incur a nontrivial amount of violated mass:
\[
    \bigl\langle \Pi_{\psi_e}+\Pi_{\psi_f},\,\rho\bigr\rangle
    \geq \Omega(1).
\]
where we let $\Pi$ be the projection matrix to the one-dimensional forbidden subspace.
Equivalently, a pairwise-uniform state cannot simultaneously achieve value arbitrarily close to one on both clauses. Thus, unlike in the classical setting, pairwise uniformity itself creates a local obstruction once two suitably positioned quantum constraints overlap.

At first sight, the rank-$1$ assumption of the forbidden direction appears too restrictive for the Haar-random model as this occurs with probability zero. The key point is that the argument is robust: the same local obstruction continues to hold whenever the forbidden directions are sufficiently close to rank-one states. Since the local Hilbert space $\C^8$ has fixed dimension, independent of $n$, a Haar-random forbidden direction has a constant probability of lying within a fixed neighborhood of the rank-one locus. Consequently, with constant probability, a random clause is sufficiently close to rank one for the bow-tie argument to apply.

Finally, ~\cref{lem:many-useful-bowties} shows that the random constraint hypergraph contains $\Omega(m)$ such usable bow-ties with high probability. Summing the resulting local violations yields an $\Omega(m)$ lower bound on the total violated mass, and hence a constant-gap strong refutation, as formalized in \cref{claim:fractional-matching-refutation}.

\paragraph{Organization}
\Cref{sec:preliminaries} fixes notation and introduces pseudo-states and the normalized objective.
\Cref{sec:proof-main-theorem} proves \cref{thm:main} by reducing to pairwise-independent clause states and separating
biased from near-uniform marginals.  \Cref{sec:almost-twise-indep} converts approximate independence into exact
independence.  \Cref{sec:bowtie-bounds} gives a refutation for close to pairwise uniform states.
The supporting conditioning and concentration results appear in \cref{sec:conditioning,sec:concentration}.


\section{Preliminaries} \label{sec:preliminaries}

\subsection{Pauli Operator}

\begin{definition}[Pauli operator]
  Recall standard Pauli matrices
  \[ X =
    \begin{bmatrix} 0 & 1 \\ 1 & 0
    \end{bmatrix} \qquad
    Y =
    \begin{bmatrix} 0 & -i \\ i & 0
    \end{bmatrix} \qquad
    Z =
    \begin{bmatrix} 1 & 0 \\ 0 & -1
  \end{bmatrix} \,. \]
  An $n$-qubit Pauli operator is a $2^n$-by-$2^n$ matrix
  $\Paren{\bigotimes_{i \in [n]} \sigma_i} \in \set{I, X, Y, Z}^{\otimes n}$.
  The weight of such a Pauli operator is the number of non-identity components $\sigma_i$.
\end{definition}

Even though this work focuses on qubits, many intermediate results can be generalized to $q$-ary qudits.
Let $\cL$ be the set of linear operators on $\C^q$.
For a system on $n$ qudits, let $\cL^n_d \defas \Span_\C \set{ \cL^W \otimes I_{[n] \setminus W} |
W \in {[n] \choose \leq d} }$ be the span of operators of degree at most $d$.

\begin{definition}[Heisenberg–-Weyl operator]
  For a $q$-ary qudit ($q \geq 2$), recall the Heisenberg--Weyl operators \cite[Definition~3.7]{khatri2024principles}
  \begin{align*}
    W_{z,x} &\defas Z(z)X(x) \qquad \text{for } 0 \leq z, x \leq q^2-1 \\
    Z(z) &\defas \sum_{0 \leq k < q} e^{2\pi ikz/q} \ket q\bra q \\
    X(x) &\defas \sum_{0 \leq k < q} \ket{k+x}\bra k \,.
  \end{align*}
  Denote the set of these $q^2$ operators by $\cW$.
  Denote by $\cW^n \defas \cW^{\otimes n}$ the set of Heisenberg–-Weyl operators on $n$ qudits, and $\cW^n_d$ those of
  weight at most $d$.
\end{definition}

The matrices $W_{z,x}$ form an orthogonal basis of $\cL$, each being unitary of Frobenius norm $\sqrt q$.
When $q = 2$, $\cW$ reduces to the Pauli matrices.

\subsection{Pseudo-Expectation}

Given a genuine state $\rho$, its expectation $\E_\rho[H] = \Tr[\rho H]$ is a linear functional that is
normalized and positive.
We call any similar functional on low-degree operators a pesudo-expectation:

\begin{definition}[Pseudo-Expectation]
  A degree-$2d$ $q$-ary $n$-qudit pseudo-expectation is a $\C$-linear functional $\pE_\rho: \cL^n_{2d} \to \C$ such
  that:
  \begin{enumerate}
    \item (Normalized) \quad $\pE_\rho [I] = 1$.
    \item (Positive) \quad $\pE_\rho [H^\dagger H] \geq 0$ for every $H \in \cL^n_d$.
  \end{enumerate}
  Given such an operator $\pE_\rho$, also say that $\rho$ is a pseudo-state.
\end{definition}

\begin{remark} \label{rem:pseudo-moment-matrix}
  A pseudo-expectation operator can be represented as a $\cW^n_d$-by-$\cW^n_d$ Hermitian matrix $M$ such that
  $M(\sigma,\tau) = \pE_\rho[\sigma^\dagger \tau]$ for $\sigma,\tau \in \cW^n_d$.
  Positivity of $\pE_\rho$ implies $M$ is positive semidefinite.
\end{remark}

\begin{definition}[Conditioning] \label{def:conditioning}
  Given any degree-$2(d+s)$ pseudo-expectation $\pE_\rho$, a POVM $\set{\cE_x}$ for a subset $T \subseteq [n]$ of $s$
  qudits (with Kraus operators $\sqrt{\cE_x}$), the pseudo-expectation $\pE_{\rho|_x}$ conditioned on the outcome $x$ is
  \begin{equation} \label{eq:conditioning}
    \pE_{\rho|_x} [H] \defas \frac{\pE_\rho [\sqrt{\cE_x} H \sqrt{\cE_x}]}{\pE_\rho [\cE_x]}
    \qquad \text{for } H \in \cL^n_{2d} \,.
  \end{equation}
\end{definition}

It is easy to show that $\pE_{\rho|_x}$ is indeed a pseudo-expectation operator of degree $2d$
(\cref{lem:conditioning}), so $\rho|_x$ is a pseudo-state.

\begin{definition}[Reduced Pseudo-State]
  Given any degree-$2d$ $n$-qudit pseudo-state $\rho$ and any subset $A \subseteq [n]$, the degree-$2d$ reduced
  pseudo-state $\rho_A$ on $A$ has pseudo-expectation
  \begin{equation} \label{eq:reduced-pseudo-state}
    \pE_{\rho_A} [H] \defas \pE_\rho [H \otimes I_{[n] \setminus A}] \qquad \text{for } H \in
    \cL^A_{2d} \,.
  \end{equation}
\end{definition}

\subsection{$t$-Wise Independent}

Next, we generalize the classical definition of $t$-wise independence to quantum states.
\begin{definition}
  A quantum state on $k$ qudits is $t$-wise $\nu$-independent if its reduced state on every subset $S$ of at most $t$
  qubits is $\nu^{\otimes |S|}$.
\end{definition}
$2$-wise $\nu$-independent is also known as pairwise $\nu$-independent.

\subsection{Random Model}

Throughout the work, we consider the binomial random graph model as in \cite{ChanDorsiXu26} and view each $k$-SAT
clause as a directed hyperedge. A directed $k$-uniform hypergraph $\cH = (V, \cE)$
has hyperedges $e \in \cE$ that are (ordered) $k$-tuples of distinct vertices.
Precisely, each potential undirected $k$-uniform hyperedge $\hat e \in {V \choose k}$ appears independently with
probability $p$, as in the \Erdos--\Renyi model of undirected hypergraphs.
Then each realized undirected hyperedge $\hat e$ will have its $k$ vertices ordered uniformly at random, from among the
$k!$ ordering.
This yields a directed hyperedge $e \in V^{\underline k}$, where $V^{\underline k}$ denotes the collection of $k$-tuples
of distinct vertices.
(The ordering of the vertices in $e$ will be useful for the analysis of our algorithm.)
Therefore $\cH$ has $p {|V| \choose k}$ hyperedges in expectation.
Each directed hyperedge $e \in V^{\underline k}$ comes with a projector constraint $\Pi_e = \ket{\psi_e} \bra{\psi_e}$,
where $\psi_e$ is a uniformly random pure state on the $k$ qubits of $e$.

For a quantum $k$-SAT instance $\cI=([n],\cE)$ with $m=|\cE|$ and forbidden projectors
$\{\Pi_e\}_{e\in\cE}$, define the normalized value of a collection of clause states $\{\phi_e\}_{e\in\cE}$ by
\[
  \val(\cI,\{\phi_e\})
  \defas \frac1m\sum_{e\in\cE}\Tr\bigl((I-\Pi_e)\phi_e\bigr)
  =1-\frac1m\sum_{e\in\cE}\Tr(\Pi_e\phi_e)
  \in[0,1].
\]
Thus $\val$ is a satisfaction objective: larger values are better, and strong refutation means certifying
$\val\leq1-\delta$ for some constant $\delta>0$.
Similarly, given a degree-$2d$ pseudo-state $\phi$ with $2d \geq k$, define $\val(\cI, \phi) = \val(\cI, \{\phi_e\})$,
where $\{\phi_e\}$ is the collection of reduced states $\phi_e$ on clauses $e$.

\section{Proof of Main Theorem} \label{sec:proof-main-theorem}

We prove the main theorem in this section, using some intermediate results that will be proved in later sections.

Given an $n$-qubit pseudo-state $\rho$ of degree (at least) $1$, its marginal $\nu(\rho)$ is the average
of the reduced state $\rho_u$ over all qubits $u \in [n]$.
Given an $n$-qubit degree-$k$ pseudo-state $\rho$ and a directed $k$-uniform hypergraph $\cH = ([n], \cE)$, the induced
state $\cE(\rho)$ is the average of the reduced state $\rho_e$ over all directed hyperedges $e \in \cE$.

The certification below only requires pair-marginal concentration at degree
$2$.  Restricting to two positions is important: a degree-$2$ certificate
controls the two-qubit marginals of the induced state, not its full $k$-qubit
density matrix.
The next definition is (the $t = 2$ special case of) \cite[Definition~4.2]{ChanDorsiXu26}, quantized.
Below $V^{\underline k}$ denotes the collection of $k$-tuples of distinct elements from $V$.

Here a degree-$2$ pseudo-state is understood in the quantum-SoS sense: its
normalized Pauli moment functional satisfies
$\widetilde{\E}_\rho[L^\dagger L]\ge0$ for every degree-$1$ global Pauli
polynomial $L$, in addition to the usual local Pauli relations and cross-site
commutation.  Equivalently,
the global degree-$1$ Pauli moment matrix is positive semidefinite.  Every
actual global state has this property; positivity of an unrelated collection
of two-qubit density matrices alone would not suffice for the spectral
certificate below.


The previous definition of pair-marginal concentration is all we need for our applications to $3$-QSAT.

\begin{definition}[Concentrated] \label{def:concentrated}
  For any $\eta > 0$, natural numbers $t$ and $q$, a $k$-uniform hypergraph $\cH = (V, \cE)$ or a $k$-uniform hyperedge
  set $\cE$ is $(\eta,q,d,t)$-concentrated if
  \[ \Norm{\E_{e \in \cE} (\rho_e)_W - \E_{e \in V^{\underline k}} (\rho_e)_W}_F \leq \eta \]
  for any degree-$d$ $q$-ary $n$-qudit pseudo-state $\rho$ and any $W \in {[k] \choose t}$.
\end{definition}

The next result is \cite[Theorem~H.15]{ChanDorsiXu26} quantized and will be proved in \cref{sec:concentration}.

\begin{restatable}{theorem}{Concentration} \label{thm:concentration}
  Fix even $t \leq k$.
  For any $\eta > 0$, any $q > 0$, there is $C = C(\eta,q,t,k)$ such that for $m \geq Cn^{t/2}$, except with probability
  $o_n(1)$, a random directed $k$-uniform hypergraph $\cH = ([n], \cE)$ with $m$ expected hyperedges is
  $(\eta,q,t,t)$-concentrated.
\end{restatable}

By \cite[Lemma~16]{BH13}, for any positive integer $q$, there is a quantum-classical channel
\begin{equation} \label{eq:quantum-classical-channel}
  \Lambda(\sigma) = \sum_{x \in [q']} \Tr[\sigma \cE_x] \ket x\bra x
\end{equation}
on a $q$-ary qudit $\sigma$ with $q' \leq q^8$ outcomes, given by a rank-$1$ POVM $\{\cE_x\}$ and an orthonormal
basis $\{\ket x\}$.
This channel satisfies for any positive integer $s$ and traceless operator $\xi$,
\begin{equation} \label{eq:distortion}
  \norm{\xi}_1 \leq (18q)^{s/2}\norm{\Lambda^{\otimes s}(\xi)}_1 \,.
\end{equation}

Given any $q$-ary $n$-qudit degree-$d$ pseudo-state $\rho$, any subset $T \subseteq [n]$ of size at most $d$, since
$\rho_T$ is a genuine state (\cref{lem:reduced-pseudo-expectation,lem:genuine-state}), there is a probability
distribution $\mu_T \defas \Lambda^{\otimes T}(\rho_T)$ of assignments $b \in [q']^T$.
Let $\cT$ be the distribution of picking a uniformly random $s \in \set{0,\dots,\kappa}$ and a uniformly random subset
$T \subseteq [n]$ of size $s$.
Given $b \in \supp(\mu_T)$, we consider the conditioned pseudo-state $\rho|_b$.

Our conditioning set $T = \set{i_1, \dots, i_s}$ has size exactly $s$, or equivalently its elements $i_1, \dots, i_s$
are chosen without replacement.
This differs from \cite[Algorithm~5.5]{BRS11} and \cite[Algorithm~24]{BH13} whose elements are chosen with replacement.
We do so to avoid complications of remeasuring the same qudits twice and getting different outcomes.

Given an $n$-qudit pseudo-state $\rho$, let $\nu(\rho) \defas \E_{i \in [n]} \rho_i$ be its average $1$-wise marginal.
The next theorem is (the $t = 2$ case of) \cite[Theorem~4.3]{ChanDorsiXu26} quantized and will be proved in
\cref{sec:conditioning} as \cref{thm:cond-indep-concentrated}.

\begin{restatable}{theorem}{ConditionalIndependence} \label{thm:cond-indep-concentrated}
  For any $\eta > 0$, there are $\kappa = \poly(q/\eta)$ and a distribution $\cT$ over ${[n] \choose \leq
  \kappa}$, such that if a directed $k$-uniform hypergraph $\cH = ([n], \cE)$ is $(\eta/(2q),q,2,2)$-concentrated and
  $n \geq \poly(q/\eta)$, then every degree-$2(\kappa+k)$ $q$-ary $n$-qudit pseudo-state $\rho$ and distinct
  $i,j \in [k]$ satisfy
  \[ \E_{T \sim \cT} \E_{b \sim \mu_T} \norm{\cE(\rho|_b)_{ij} - \nu(\rho|_b)^{\otimes 2}}_1
  \leq \eta \,. \]
  That is, in expectation over subset $T \sim \cT$ and outcome $b \sim \mu_T = \Lambda^{\otimes T}(\rho_T)$, the
  $ij$-marginal of the conditional induced state $\cE(\rho|_b)$ is close to $\nu(\rho|_b)^{\otimes 2}$.
  Further, for every subset $S$ of size $k \geq 2$,
  \[ \norm{\E_{T,\, b} (\rho|_b)_S - \rho_S}_1 \leq \kappa k/n \,. \]
\end{restatable}

The next result is the quantum version of \cite[Theorem~4.11]{ChanDorsiXu26} and \cite[Theorem~B.15]{ACLM26}, which
generalize \cite[Theorem~2.1]{alon03almost}.
It will be proved in \cref{sec:almost-twise-indep} as \cref{thm:almost-twise-indep}.
Recall the trace distance between two states $\rho$ and $\tau$ is
\[ D(\rho, \tau) = \frac 12 \norm{\rho - \tau}_1 \,. \]

\begin{restatable}{theorem}{AlmostTwise} \label{thm:almost-twise-indep}
  Consider $q$-ary qudit systems.
  Let $\nu$ be a single-qudit state $\nu$ and $\rho$ a $k$-qudit state.
  Consider
  \begin{equation} \label{eq:almost-twise-indep}
    \eps \defas \sum_{W \in {[k] \choose \leq t}} (q^2-1)^{|W|} \norm{\rho_W - \nu^{\otimes W}}_1 \,.
  \end{equation}
  Then there is a $t$-wise $\nu$-independent state $\tau$ such that
  \[ D(\rho, \tau) \lesssim_{k,q} \eps^{1/(2k+2)} \,. \]
\end{restatable}

\begin{corollary} \label{cor:cond-2wise-indep}
  For any $\eta > 0$, there are $d = \poly(q/\eta)$ and a distribution $\cT$ over ${[n] \choose \leq d}$, such that if
  a directed $k$-uniform hypergraph $\cH = ([n], \cE)$ is $(\Omega(\eta/q)^{2k+2},q,2,2)$-concentrated and
  $n \geq \poly(q/\eta)$, then for every degree-$2(d+k)$ $q$-ary $n$-qudit pseudo-state $\rho$,
  \[ \E_{T \sim \cT} \E_{b \sim \Lambda^{\otimes T}(\rho_T)} \norm{\cE(\rho|_b) - \phi_b}_1 \leq \eta \,, \]
  where $\phi_b$ is a $\nu(\rho|_b)$-independent state for every subset $T \in \supp(\cT)$ and outcome
  $b \in \supp(\Lambda^{\otimes T}(\rho_T))$.
  Further, for every subset $S$ of size $k \geq 2$,
  \[ \norm{\E_{T,\, b} (\rho|_b)_S - \rho_S}_1 \leq dk/n \,. \]
\end{corollary}

\begin{proof}
  We will apply \cref{thm:almost-twise-indep}.
  Note that $\eps_\emptyset = 0$ and we may assume $\eps_U \leq \eps_W$ for $U \subseteq W$.
  It suffices to bound $\eps_W$ for subsets $W$ of size exactly $2$.
  \cref{thm:cond-indep-concentrated} implies
  \[ \E_{T \sim \cT} \E_{b \sim \Lambda^{\otimes T}(\rho_T)} \sum_{W \in {[k] \choose \leq 2}}
  \norm{\cE(\rho|_b)_W - \nu(\rho|_b)^{\otimes W}}_1 \leq (\eta/q)^{2k+2} \,. \]
  For $T \in \supp(\cT), b \in \supp(\Lambda^{\otimes T}(\rho_T))$, let $\phi_b$ be the pairwise
  $\nu(\rho|_b)$-independent distribution closest to $\cE(\rho|_b)$.
  The result follows by \cref{thm:almost-twise-indep}.
\end{proof}

\subsection{$\eps$-Net of Constraints}

The induced state $\cE(\rho)$ considered in the previous subsection is an average over all constraints in $\cE$.
For our proof, it is helpful to split the constraints into groups, where all constraints in the same group forbids
pure states that are close to each other.
For this reason, we will consider an $\eps$-net over pure product states on $k$ qubits, because our refutation are based
on such constraints.
In fact, a constraint forbids a pure product states with probability $0$.
So we will also consider pure states that are $\delta$-close to a pure product state.
We will then consider the induced state $\cE_{\tilde\psi}(\rho)$ that is the average over all constraints forbidding a
state close to $\tilde\psi$.
Doing so will help us strongly refute pseudo-states whose marginal $\nu$ is biased in \cref{sec:biased-marginal}.

Given $k \geq 1$ and $\eps>0$, let $\cN_\eps$ be an $\eps$-net of the pure states on $k$ qubits of cardinality
$O(1/\eps)^{2^{k+1}}$, obtained by taking any maximal set whose points have pairwise distance at least $\eps$.
The upper bound on cardinality is due to pure states lying in the unit sphere of $\C^{2^k}$, which is isometric to
the unit sphere of $\R^{2^{k+1}}$.
Given a pure state $\psi$ on $k$ qubits, let $\cN_\eps(\psi)$ be the pure state in $\cN_\eps$ closest to $\psi$,
breaking ties arbitrarily.

\begin{definition} \label{def:eps-net-instance}
  Given a $k$-QSAT instance $\cI$ and an $\eps$-net $\cN_\eps$ as above, denote by $\tilde\cI$ the instance
  obtained from $\cI$ by replacing the forbidden state $\psi$ of every constraint with the closest state $\tilde\psi
  \defas \cN_\eps(\psi)$ in the $\eps$-net.
  Denote by $\tilde\Psi$ the multiset of forbidden states of the constraints in $\tilde\cI$.
  Given $\tilde\psi \in \tilde\Psi$, let $\cE_{\tilde\psi}$ be the collection of constraints in $\tilde\cI$ whose
  forbidden state is $\tilde\psi$.
  The instance $\cI$ is $\eta$-concentrated if the directed multihypergraph of
  scopes $\cE_{\tilde\psi}$ is $\eta$-concentrated for every
  $\tilde\psi\in\tilde\Psi$.
\end{definition}

\begin{lemma} \label{lem:cond-2wise-indep-multi}
  For any $\eta > 0$, there are $d > 0$ and a distribution $\cT$ over ${[n] \choose \leq d}$, such that
  for any $k$-QSAT instance $\cI = ([n], \cC)$, if $\cE_{\tilde\psi}$ is $(\Omega(\eta/q)^{2k+2},2,2,2)$-concentrated
  for $\tilde\psi \in \tilde\Psi$, then for every degree-$2(d+k)$ $n$-qubit pseudo-state $\rho$,
  \[ \E_{T \sim \cT} \E_{b \sim \Lambda^{\otimes T}(\rho_T)} \E_{\tilde\psi \in \tilde\Psi}
  \Norm{\cE_{\tilde\psi}(\rho|_b) - \phi_{b,\tilde\psi}}_1 \leq \eta \,, \]
  where $\phi_{b,\tilde\psi}$ is a pairwise $\nu(\rho|_b)$-independent state for subset $T \in \supp(\cT)$, outcome
  $b \in \supp(\Lambda^{\otimes T}(\rho_T))$, constraint type $\tilde\psi \in \cN_\eps$.
  Further, for every subset $S$ of size $k \geq 2$,
  \begin{equation} \label{eq:approx-total-prob-multi}
    \norm{\E_{T,\, b} (\rho|_b)_S - \rho_S}_1 \leq dk/n \,.
  \end{equation}
\end{lemma}

\begin{proof}
  Apply \cref{cor:cond-2wise-indep} to each $\cE_{\tilde\psi}$.
\end{proof}

Crucially, the states $\phi_{b,\tilde\psi}$ from the previous lemma are pairwise $\nu$-independent, where
$\nu = \nu(\rho|_b)$ is unrelated to $\tilde\psi$.
This fact will be useful for the strong refutation for biased $\nu$ in \cref{lem:biased-marginal}.

\subsection{Biased Marginal} \label{sec:biased-marginal}

The next lemma shows that if the induced states $\set{\phi_\psi}$ are $\nu$-independent for some $\gamma$-biased $\nu$,
then they $\Omega(\gamma)$-violate the constraints in $\tilde\cI$, provided the forbidden states $\Psi$ of $\cI$ are
approximately evenly distributed (\cref{def:smooth}):

\begin{lemma} \label{lem:biased-marginal}
  For any $\gamma > 0$, there is $\beta_{\ref{lem:biased-marginal}} > 0$ such that for any collection
  $\set{\phi_{\tilde\psi} | \tilde\psi \in \cN_\eps}$ of pairwise $\nu$-independent states on $3$ qubits, if
  $\norm{\nu-\frac I2}_F \geq \gamma$ and $\Psi$ is $\tau$-smooth,
  \[ \val(\tilde\cI, \set{\phi_{\tilde\psi}}) \leq 1-\beta_{\ref{lem:biased-marginal}} + o_{\gamma;\tau,\eps}(1)
    \,.\footnote{$o_{\gamma;\tau,\eps}(1)$ represent a quantity that, for every fixed $\gamma > 0$, is arbitrarily close
  to $0$ when $\max(\tau,\eps)$ is sufficiently close to $0$.} \]
\end{lemma}

\begin{proof}
  Expand $\nu = \lambda\ket\theta\bra\theta + (1-\lambda)\ket{\theta^\perp}\bra{\theta^\perp}$ in its orthonormal
  eigenbasis, where $\lambda > \frac12$.
  Consider the pure product state $\theta^3$.
  \cref{lem:biased-marginal-constraint} implies any pairwise $\nu$-independent state $\phi$ satisfies
  $\bra{\theta^3}\phi \ket{\theta^3} \geq 4c \gamma$ for some $c > 0$.
  The spherical cap $C(\theta^3,t) \subset \S^{2^k-1}(\C)$ of pure states $\chi$
  such that $\norm{\theta^3-\chi} \leq c\gamma$ has measure $\alpha > 0$ under the uniform distribution of
  $\S^{2^k-1}(\C)$.
  Since $\Psi$ is $\tau$-smooth, $\Pr_{\psi \in \Psi}[\psi \in C(\theta^3,t)] \geq \alpha - \tau$.
  For $\psi \in C(\theta^3,t)$, \cref{lem:outer-product-diff-norm} implies
  $\bra{\tilde\psi}\phi_{\tilde\psi} \ket{\tilde\psi} \geq 2c\gamma-2\eps$, where
  $\tilde\psi \defas \cN_\eps(\psi)$.
  Therefore
  \begin{align*}
    1-\val(\tilde\cI,\set{\phi_{\tilde\psi}})
    &= \E_{\psi \in \Psi} \bra{\tilde\psi}\phi_{\tilde\psi} \ket{\tilde\psi}
    \geq \E_{\psi \in \Psi} \Iv{\psi \in C(\theta^3,t)} \cdot \bra{\tilde\psi}\phi_{\tilde\psi} \ket{\tilde\psi} \\
    &\geq (\alpha-\tau)(2c\gamma-2\eps) \geq 2c\alpha\gamma - o_{\gamma;\tau,\eps}(1) \,. \qedhere
  \end{align*}
\end{proof}

\begin{lemma} \label{lem:biased-marginal-constraint}
  Let $\nu$ be a state on a single qubit such that $\norm{\nu - \frac I2}_F \geq \gamma > 0$, and $\theta$ be the
  maximum eigenvector of $\nu$.
  Any pairwise $\nu$-independent state $\phi$ on $3$ qubits satisfies $\bra{\theta^3}\phi\ket{\theta^3} \geq
  \gamma/\sqrt 2$.
\end{lemma}

\begin{proof}
  Without loss of generality, assume $\ket 0$ and $\ket 1$ are the largest and smallest eigenvectors of $\nu$, by
  redefining $\ket 0$ and $\ket 1$ as such for the rest of this proof.
  Also let $\lambda$ be the largest eigenvalue of $\nu$.
  Then $\gamma^2 \leq \norm{\nu - \frac I2}_F^2 = 2(\lambda-\frac12)^2$, i.e.~$\lambda
  \geq  \frac12 + \frac\gamma{\sqrt 2}$.

  For $b \in \set{0,1}^3$, let $e_b \defas \bra b\phi\ket b$.
  Each $e_b \geq 0$ because $\phi$ is PSD.
  Since $\phi$ is pairwise $\nu$-independent, tracing out qubit $3$ yields
  \[ \lambda^2 = \bra{00}\nu^{\otimes 2}\ket{00} = \bra{00}\Tr_3(\phi)\ket{00} = e_{000} + e_{001}\,. \]
  Similarly, tracing out qubit $1$ and qubit $2$ respectively,
  \[ \lambda^2 = e_{000} + e_{001} = e_{000} + e_{010} = e_{000} + e_{100} \,. \]
  In particular, $e_{001} = e_{010} = e_{100}$.
  Again tracing out qubit $3$,
  \[ \lambda(1-\lambda) = \bra{01}\nu^{\otimes 2}\ket{01} = \bra{01}\Tr_3(\phi)\ket{01} = e_{010} + e_{011}\,. \]
  Similarly, rotating the roles of the three qubits in a cycle,
  \begin{equation} \label{eq:middle-weight}
    \lambda(1-\lambda) = e_{010} + e_{011} = e_{001} + e_{101} = e_{100} + e_{110} \,.
  \end{equation}
  This further implies $e_{011} = e_{101} = e_{110}$.
  Therefore $e_b$ only depends on the Hamming weight of $b$.
  Tracing out qubit $3$ once more,
  \begin{equation} \label{eq:upper-bound-weight-3}
    (1-\lambda)^2 = \bra{11}\nu^{\otimes 2}\ket{11} = \bra{11}\Tr_3(\phi)\ket{11} = e_{110} + e_{111}
    \geq e_{111} \,.
  \end{equation}

  We have $\sum_{b \in \set{0,1}^3} e_b = \Tr(\phi) = 1$, so
  \begin{align*}
    \bra{000}\phi\ket{000}
    &= e_{000} = 1 - \sum_{b \neq 000} e_b
    \overset{\eqref{eq:middle-weight},\eqref{eq:upper-bound-weight-3}}\geq 1-3\lambda(1-\lambda)-(1-\lambda)^2 \\
    &= \lambda(2\lambda-1) \geq \gamma/\sqrt 2 \,. \qedhere
  \end{align*}
\end{proof}

A spherical cap of the complex $d$-sphere $\S^d(\C) \subset \C^{d+1}$ is
\[ C(x,t) \defas \set{ y \in \S^d(\C) | \Re \inner{x,y} \geq t }\,.  \]

We test whether a multiset of points is evenly distributed using spherical caps:
\[ \sC_d \defas \set{ C(x,t) | x \in \S^d, t \in [-1,1] } \,. \]

\begin{definition}[Smooth] \label{def:smooth}
  A multiset $Y$ of points in $\S^d(\C)$ is $\tau$-smooth if
  \[ \sup_{S \in \sC_d} \abs{\sigma(S) - \mu(S)} \leq \tau\,, \]
  where $\sigma$ is the uniform distribution on $\S^d(\C)$, and $\mu(S) \defas \abs{S \cap Y}/\abs{Y}$ is the empirical
  distribution of $Y$.
\end{definition}

A spherical cap of the complex $d$-sphere $\S^d(\C)$ is a spherical cap of the real $(2d+1)$-sphere $S^{2d+1}(\R)$, and
vice versa, via the invertible isometry
\begin{equation} \label{eq:isometry}
  (\alpha_1 + \beta_1 i, \dots, \alpha_{d+1} + \beta_{d+1} i) \in \C^{d+1} \quad
  \mapsto \quad (\alpha_1, \beta_1, \dots, \alpha_{d+1}, \beta_{d+1}) \in \R^{2d+2} \,.
\end{equation}
Spherical caps of $\S^d(\C)$ therefore have the same VC dimension of $2d+3$ as those of $\S^{2d+1}(\R)$
\cite[Proposition~5.12]{BL15}.

For any fixed dimension $d$, the VC inequality (\cite[Proof of Theorem~6.7]{shalev2014understanding} or
\cite[p.~31]{devroye2001combinatorial}) implies $m$ random points $Y$ in $\S^d$ is $o_m(1)$-smooth except with
probability $o_m(1)$.

%
%
%
%
\subsection{Wrapping Up}

\MainTheorem*

\begin{proof}
  Fix $\gamma>0$ sufficiently small for \cref{claim:fractional-matching-refutation}, and let
  $\beta_{\mathrm{bias}}>0$ be the constant supplied by \cref{lem:biased-marginal} with parameter
  $\gamma/\sqrt2$.  Set $\beta_*\defas\min\{\beta_{\mathrm{bias}},\beta_{\mathrm{unif}}(\gamma)\}$, choose
  $\eps,\eta,\tau>0$ below, and take $D$ large enough to apply \cref{cor:cond-2wise-indep} classwise while retaining
  degree $4$.
  Work on the simultaneous high-probability event of \cref{claim:fractional-matching-refutation}, $\Psi$ is
  $\tau$-smooth, and each $\cE_{\tilde\psi}$ is $(\Omega(\eta^8),2,2,2)$-concentrated.

  \begin{align*}
    \abs{\val(\cI, \rho) - \E_{T,b} \val(\cI, \rho|_b)}
    &= \abs{\E_{(e,\psi) \in \cC} \E_{T,b} \bra \psi ((\rho|_b)_e - \rho_e) \ket \psi} \\
    &\leq \E_{(e,\psi) \in \cC} \norm{\E_{T,b} (\rho|_b)_e - \rho_e}_1
    \overset{\eqref{eq:approx-total-prob-multi}}= o_n(1) \,.
  \end{align*}
  Now for any degree-$4$ pseudo-state $\sigma$,
  \begin{align*}
    \abs{\val(\cI, \sigma) - \val(\tilde\cI, \sigma)}
    &= \abs{\E_{(e,\psi) \in \cC} \Tr[\sigma_e (\ket\psi\bra\psi - \ket{\tilde\psi}\bra{\tilde\psi})]} \\
    &\leq \E_{(e,\psi) \in \cC} \norm{\ket\psi\bra\psi - \ket{\tilde\psi}\bra{\tilde\psi}}_\infty \leq 2\eps \,,
  \end{align*}
  using \cref{lem:outer-product-diff-norm}.
  Also
  \begin{align*}
    \val(\tilde\cI, \sigma) &= 1-\E_{(e,\psi) \in \cC} \bra{\tilde\psi} \sigma_e \ket{\tilde\psi}
    = 1-\E_{\tilde\psi \in \tilde\Psi} \bra{\tilde\psi} \cE_{\tilde\psi}(\sigma) \ket{\tilde\psi} \\
    &= \val(\tilde\cI, \{\cE_{\tilde\psi} (\sigma)\}) \,.
  \end{align*}
  \cref{lem:cond-2wise-indep-multi} implies there are pairwise $\nu(\rho|_b)$-independent states $\phi_{b,\tilde\psi}$
  such that
  \begin{align*}
    \abs{\val(\tilde\cI, \{\cE_{\tilde\psi} (\rho|_b)\}) - \val(\tilde\cI, \{\phi_{b,\tilde\psi} \})}
    &= \abs{\E_{\tilde\psi \in \tilde\Psi} \Tr[(\cE_{\tilde\psi} (\rho|_b) - \phi_{b,\tilde\psi})
    \ket{\tilde\psi}\bra{\tilde\psi}]} \\
    &\leq \E_{\tilde\psi \in \tilde\Psi} \norm{\cE_{\tilde\psi}(\rho|_b) - \phi_{b,\tilde\psi} }_1 \,,
  \end{align*}
  which is at most $\eta$ in expectation over $T \sim \cT, b \sim \Lambda^{\otimes T} (\rho_T)$.
  The previous four display equations and triangle inequality imply
  \[ \abs{\val(\cI, \rho) - \E_{T,b} \val(\tilde\cI, \{\phi_{b,\tilde\psi}\})} \leq 2\eps + o_n(1) + \eta\,. \]
  If $\norm{\nu(\rho|_b) - \frac I2}_F \geq \gamma/\sqrt 2$, then \cref{lem:biased-marginal} implies
  $\val(\tilde\cI, \{\phi_{b,\tilde\psi}\}) \leq 1 - \beta_{\mathrm{bias}} + o_{\gamma;\tau,\eps}(1)$.
  If $\norm{\nu(\rho|_b) - \frac I2}_F \leq \gamma/\sqrt 2$, then \cref{claim:fractional-matching-refutation} implies
  $\val(\tilde\cI, \{\phi_{b,\tilde\psi}\}) \leq 1 - \beta_{\mathrm{unif}}(\gamma) + 2\eps$.

  Choose $\Delta$ sufficiently large, then
  $\eps,\tau,\eta$ sufficiently small, and finally $n$ sufficiently large.  Thus, taking $\delta=\beta_*/2$, every
  feasible degree-$D$ pseudo-state has value at most $1-\delta$.  Since $D$ is constant for fixed $\Delta$, the
  degree-$D$ quantum-SoS relaxation can be solved in polynomial time, and its optimum upper bounds the value of every
  genuine quantum state.  Solving this relaxation therefore gives the claimed deterministic strong-refutation
  algorithm.
\end{proof}

\section{Almost $t$-Wise Independent} \label{sec:almost-twise-indep}

\AlmostTwise

Following \cite[Theorem~4.11]{ChanDorsiXu26} and \cite[Theorem~B.15]{ACLM26}, we first prove
\cref{thm:almost-twise-indep} for the special case when $\nu$ is somewhat unbiased, by quantizing the short proof
of \cite[Lemma~B.16]{ACLM26}.

\begin{lemma} \label{lem:almost-twise-indep-unbiased}
  In the context of \cref{thm:almost-twise-indep}, suppose $\lambda_{\min}(\nu) \geq \gamma/q$.
  Then there is a $t$-wise $\nu$-independent state $\tau$ such that $D(\rho, \tau) \leq 2 \gamma^{-k} \eps$.
\end{lemma}

\begin{proof}
  Let $\Delta = \rho - \nu^{\otimes k}$.
  Expand $\Delta$ in the Heisenberg–-Weyl basis $\cW^k$ on $k$ qudits,
  \[ \Delta = \frac 1{q^k} \sum_{G \in \cW^k} \Tr[\Delta G^\dagger] G\,. \]
  For any Heisenberg--Weyl operator $G = \bigotimes_{i \in [k]} \sigma_i \in \cW^k$ whose non-identity components
  $\{ i \in [k] \mid \sigma_i \neq I \}$ are precisely $W$, we have $\Abs{\Tr \Delta G^\dagger} \leq \norm{\Tr_{[k]
  \setminus W} \Delta}_1 \norm{G}_\infty = \norm{\rho_W - \nu^{\otimes W}}_1$.

  Construct $\tau$ by removing from $\rho$ all terms of degree at most $t$, and rescaling to ensure positivity:
  \begin{equation} \label{eq:high-deg-perturb}
    \tau \defas \nu^{\otimes k} + \frac{1-\delta}{q^k} \sum_{G \in \cW^k \setminus \cW^k_t} \Tr[\Delta G^\dagger] G \,,
  \end{equation}
  where $\delta \defas \eps\gamma^{-k}$.

  Then $\Tr \tau = \Tr \nu^{\otimes k} = 1$ because non-identity Heisenberg--Weyl operators are traceless.
  Also $\tau^\dagger = \tau$ because $\cW^k \setminus \cW^k_t$ is invariant under taking adjoint.
  Rewriting
  \[ \tau = \delta \nu^{\otimes k} + (1-\delta) \rho - \frac{1-\delta}{q^k}
  \sum_{G \in \cW^k_t} \Tr[\Delta G^\dagger] G \,, \]
  we have
  \begin{align*}
    \lambda_{\min}(\tau)
    &\geq \delta \lambda_{\min}(\nu^{\otimes k})
    - \frac{1-\delta}{q^k} \sum_{G \in \cW^k_t} \abs{\Tr \Delta G^\dagger} \norm{G}_\infty \\
    &\geq \frac1{q^k} \Paren{\delta \gamma^k - (1-\delta) \eps } \geq 0 \,,
  \end{align*}
  where we have used $\lambda_{\min}(\nu) \geq \gamma/q$, $\delta = \eps\gamma^{-k}$, and \cref{eq:almost-twise-indep}.
  Therefore $\tau$ is a density matrix.

  For any subset $W \in {[k] \choose \leq t}$ and any Heisenberg--Weyl $G \in \cW^k \setminus \cW^k_t$, the partial
  trace $\Tr_{\overline W} G$ over the complement $\overline W = [k] \setminus W$ is zero.
  \cref{eq:high-deg-perturb} implies $\Tr_{\overline W} \tau = \Pr_{\overline W} \nu^{\otimes k} = \nu^{\otimes W}$, so
  $\tau$ is $t$-wise $\nu$-independent.

  Finally,
  \[ \rho - \tau = \delta \Delta + \frac{1-\delta}{q^k} \sum_{G \in \cW^k_t} \Tr[\Delta G^\dagger] G \,. \]
  In terms of Schatten $1$-norm, the first term on the right is at most $2\delta$, and the second is at most
  $\eps$.
  Therefore $D(\rho, \tau) \leq (2\delta + \eps)/2 \leq 2\gamma^{-k} \eps$ because $\gamma \leq 1$.
\end{proof}

We now prove \cref{thm:almost-twise-indep} without assuming the bias of $\nu$, by first massaging $\nu$ into a somewhat
unbiased $\tilde\nu$ so that \cref{lem:almost-twise-indep-unbiased} applies.
Our proof is a quantized variant of \cite[Theorem~4.11]{ChanDorsiXu26} and \cite[Theorem~B.15]{ACLM26} that fuse light
probability masses into heavy ones.
To redistribute the mass, our proof uses quantum channels $\Phi^{\otimes k}$ and $\Psi^{\otimes k}$, whose design and
analysis are more delicate than the classical setting.

\begin{proof}[Proof of \cref{thm:almost-twise-indep}]
  Let $\gamma \defas \eps^{1/(k+1)}$.
  We will assume $t \geq 1$, because the theorem is vacuous for $t = 0$.
  Let $H$ be the span of the eigenvectors of $\nu$ with eigenvalue at least
  $\gamma$, $L = H^{\perp}$, $d_H \defas \dim H$, and $\Pi_H, \Pi_L$ the projections to $H, L$ respectively.
  We will also assume $\Tr [\nu\Pi_L] > 0$, for otherwise the theorem follows from
  \cref{lem:almost-twise-indep-unbiased}.

  Consider the single-qudit channel
  \[ \Phi(\sigma) \defas \Pi_H\sigma\Pi_H + \frac{\Pi_H}{d_H}\,\Tr[\sigma\Pi_L]\,,
  \]
  which keeps the heavy part of $\sigma$ --- including all heavy--heavy coherence --- and sweeps the light mass
  $\Tr[\Pi_L\sigma]$ evenly onto $H$.
  $\Phi$ preserves trace because $\mathrm{Tr}\,\Phi(\sigma) = \mathrm{Tr}[\sigma\Pi_H] + \mathrm{Tr}[\sigma\Pi_L]
  = \mathrm{Tr}\,\sigma$.
  $\Phi$ is also completely prositive, hence a channel.
  We set $\tilde\nu \defas \Phi(\nu)$ and $\tilde\rho \defas \Phi^{\otimes k}(\rho)$.

  We now show that $\tilde\rho$ is close to $\rho$.
  We have
  \[ \alpha_\perp \defas \Tr[\rho(I - \Pi_H^{\otimes k})] \leq \sum_{i \in [k]} \Tr[\rho \Pi^{(i)}_L]
  \overset{\eqref{eq:almost-twise-indep}}\leq \eps + \sum_{i \in [k]} \Tr[\nu \Pi_L] \leq \eps + kq\gamma. \]
  The first inequality is due to $I - \Pi_H^{\otimes k} \preccurlyeq \sum_{i=1}^k \Pi_L^{(i)}$ and
  $\rho \succcurlyeq 0$; the middle inequality is because the $i$-coordinate heavy projection $\Pi_H^{(i)} \in \cL^k_1$
  (and hence \cref{eq:almost-twise-indep} applies); the last inequality uses the fact that every light eigenvalue of
  $\nu$ is less than $\gamma$ and there are at most $q$ of them.
  The all-heavy Kraus operator of $\Phi^{\otimes k}$ is $\Pi_H^{\otimes k}$, and \cref{claim:fuchs-vandegraaf} yields
  \[ D(\rho,\tilde\rho) \leq \sqrt{1 - (1 - \alpha_\perp)^2} \leq \sqrt{2\alpha_\perp}
  \leq \sqrt{2(kq\gamma + \eps)} \lesssim_{k,q} \sqrt\gamma \,. \]

  Next, we show $\tilde\rho$ is almost $t$-wise $\tilde\nu$-independent.
  Using $\Phi^{\otimes k}(\nu^{\otimes k}) = \tilde\nu^{\otimes k}$, for every $W \subseteq [k]$,
  \cref{lem:partial-trace-channel} implies
  \[ \Tr_{\overline W} (\tilde\rho - \tilde\nu^{\otimes k}) = \Phi^{\otimes W} (\rho_W - \nu^{\otimes W}) \,. \]
  $\Phi$ and hence $\Phi^{\otimes W}$ is a trace-norm contraction.
  Thus
  \[ \norm{\Tr_{\overline W} (\tilde\rho - \tilde\nu^{\otimes k})}_1 \leq \norm{\rho_W - \nu^{\otimes W}}_1 \,. \]
  Summing over all $W \in {[k] \choose \leq t}$,
  \[ \tilde\eps \defas \sum_{W \in {[k] \choose \leq t}} (q^2-1)^{|W|} \norm{\tilde\rho_W - \tilde\nu^{\otimes W}}_1
    \leq \sum_{W \in {[k] \choose \leq t}} (q^2-1)^{|W|} \norm{\rho_W - \nu^{\otimes W}}_1
  \overset{\eqref{eq:almost-twise-indep}}= \eps \,. \]

  We have $\tilde\nu = \Pi_H\nu\Pi_H + \frac{\Tr[\Pi_L\nu]}{d_H}\Pi_H \succcurlyeq{} \gamma\,\Pi_H$, so $\tilde\nu$
  has minimum eigenvalue at least $\gamma$ on its support $H$.
  Both $\tilde\rho$ and $\tilde\nu^{\otimes k}$ are supported on $H^{\otimes k}$,
  so \cref{lem:almost-twise-indep-unbiased} restricted to the subspace $H^{\otimes k}$ produces a $t$-wise
  $\tilde\nu$-independent state $\tilde\tau$ with $D(\tilde\rho,\tilde\tau) \leq 2\gamma^{-k}\tilde\eps \leq 2\gamma\,$.

  To recover $\nu$-independence, we now undo the mass transfer of $\Phi$ with a channel $\Psi$ satisfying
  $\Psi(\tilde\nu)=\nu$, and show that it moves $\tilde\tau$ by only $O_{k,q}(\sqrt\gamma)$.
  Write $\alpha \defas \Tr[\Pi_L\nu]$, $\nu_H \defas \Pi_H\nu\Pi_H$, and $\nu_L \defas \Pi_L\nu\Pi_L$, so that
  $\nu = \nu_H + \nu_L$ and $\tilde\nu = \nu_H + \frac\alpha{d_H}\Pi_H$.
  Define the single-qudit channel
  \[ \Psi(\sigma) \defas A_H\sigma A_H^\dagger + \frac{\nu_L}{\Tr \nu_L} \Tr[\sigma(I-A_H^2)] \,, \]
  where $A_H \defas \sqrt{\nu_H\,\tilde\nu^{-1}} = \sqrt\nu\,\tilde\nu^{-1/2}\Pi_H$ is the ``heavy'' Kraus operator,
  which is positive semidefinite by commutivity of $\nu_H$ and $\tilde\nu$.
  $\Psi$ preserves trace because $\Tr \Psi(\sigma) = \Tr [\sigma A_H^2] + \Tr [\sigma (I - A_H^2)] = \Tr\sigma$.
  $\Psi$ is also completely positive, hence a channel.
  Also, $\Psi(\tilde\nu) = \nu_H + \nu_L = \nu$.

  Now $\tau \defas \Psi^{\otimes k}(\tilde\tau)$ is $t$-wise $\nu$-independent:
  For any $W\in {[k]\choose \leq t}$, its marginal on $W$ is
  $\Tr_{\overline W}\tau = \Psi^{\otimes W}(\Tr_{\overline W} \tilde\tau) = \Psi^{\otimes W}(\tilde\nu^{\otimes W})
  = \Psi(\tilde\nu)^{\otimes W} = \nu^{\otimes W}$, where the first equality uses \cref{lem:partial-trace-channel} and
  trace-preservation of $\Psi$, and the second equality uses $t$-wise $\tilde\nu$-independence of $\tilde\tau$.

  It remains to show that $\tau$ is close to $\tilde\tau$; this mirrors the bound on $D(\rho,\tilde\rho)$ above.
  Since $\nu_H$ and $\tilde\nu$ commute and $\tilde\nu \succcurlyeq \nu_H$, we have $\sqrt{\nu_H\tilde\nu}
  \succcurlyeq \nu_H$.
  Hence
  \[ \Tr[\tilde\nu(I - A_H)] = 1 - \Tr\sqrt{\nu_H\tilde\nu} \le 1 - \Tr\nu_H = \alpha \le q\gamma\,. \]
  Denote by $(I-A_H)^{(i)} = I^{\otimes (i-1)} \otimes (I - A_H) \otimes I^{\otimes (k-i)}$. Then
  \[ \Tr[\tilde\tau(I - A_H^{\otimes k})] \leq \sum_{i \in [k]} \Tr[\tilde\tau(I-A_H)^{(i)}]
  = \sum_{i \in [k]} \Tr[\tilde\nu (I-A_H)] \leq k\alpha \,, \]
  where the first inequality uses $I - A_H^{\otimes k} \preccurlyeq \sum_{i \in [k]} (I-A_H)^{(i)}$ and
  $\tilde\tau \succcurlyeq 0$, and the equality uses $t$-wise $\tilde\nu$-independence of $\tilde\tau$.
  Therefore
  \[ D(\tilde\tau,\tau) \leq \sqrt{1-(1-k\alpha)^2} \leq \sqrt{2k\alpha} \lesssim_{k,q} \sqrt\gamma \,. \]

  Combining,
  \[ D(\rho,\tau) \le D(\rho,\tilde\rho) + D(\tilde\rho,\tilde\tau) + D(\tilde\tau,\tau)
  \lesssim_{k,q} \sqrt\gamma + \gamma + \sqrt\gamma \lesssim \sqrt\gamma \,. \qedhere \]
\end{proof}

\begin{claim} \label{claim:fuchs-vandegraaf}
  For any state $\rho$ and quantum channel $\Phi$ with Kraus operators $K_j$,
  \[ D(\rho,\Phi(\rho)) \leq \sqrt{1-\textstyle \sum_j \abs{\Tr[K_j \rho]}^2} \,. \]
\end{claim}

\begin{proof}
  Let $\psi$ be a purification of $\rho$ on a reference system $R$, and $\sigma_{RA} = (I_R \otimes \Phi)
  (\ket\psi\bra\psi)$ be the joint output state.
  Since $\psi$ is pure, its Uhlmann fidelity $F(\ket\psi\bra\psi, \sigma_{RA})$ with $\sigma_{RA}$ equals the overlap
  $\bra\psi\sigma_{RA}\ket\psi$.
  The last expression is Schumacher entanglement fidelity and equals $\sum_j \abs{\Tr[K_j \rho]}^2$
  \cite[Eq.~(10)]{schumacher1996sending}, where $K_j$ are the Kraus operators of $\Phi$.
  Monotonicity of trace distance and Fuchs--van de Graaf inequality imply
  \[ D(\rho,\Phi(\rho)) \leq D(\ket\psi\bra\psi, \sigma_{RA}) \leq \sqrt{1-F(\ket\psi\bra\psi, \sigma_{RA})}
  = \sqrt{1-\textstyle \sum_j \abs{\Tr[K_j \rho]}^2} \,. \qedhere \]
\end{proof}

\section{Strong Refutation from Bow-Ties}
\label{sec:bowtie-bounds}

\renewcommand{\theHlemma}{\thelemma}
\renewcommand{\theHdefinition}{\thedefinition}


In this section, we formalize the observation that a suitable \emph{bow-tie}---a pair of quantum $3$-SAT clauses
whose scopes intersect in exactly one qubit---forces a local violation for pairwise-uniform states.
The argument has three parts:
\begin{enumerate}
	\item In \cref{sec:exact-product-bowties}, we adapt the dominating-polynomial viewpoint from classical CSPs to a
	  bow-tie whose forbidden projectors are rank-one products.  The sum of the two clause projectors dominates an
	  explicit degree-$2$ operator.  Noncommutation on the shared qubit makes its constant term strictly positive, so
	  every state on the union whose within-clause one- and two-qubit marginals are maximally mixed incurs positive
	  violation on the pair.
	\item In \cref{sec:useful-bowties}, we transfer this inequality to rank-one forbidden projectors that are close to
	  product projectors and prove that a random instance contains many such \emph{useful} bow-ties.  Although a Haar-random
	  forbidden state is exactly product with probability zero, any fixed neighborhood of the product locus has
	  positive probability in fixed dimension.
	\item In \cref{sec:near-uniform-refutation}, we place the useful bow-ties in a fractional matching and aggregate
	  their local inequalities.  This gives a constant violation for states close to pairwise uniform, yielding a strong refutation. \end{enumerate}

\subsection{Refutation for Rank-One  Bow-Ties}
\label{sec:exact-product-bowties}

\begin{lemma}[Exact product bow-tie bound]
  \label{lem:exact-product-bowtie}
  Let $e,f$ be two quantum $3$-SAT clauses on variables $I_e, I_f \subseteq [n]$ with $I_e\cap I_f=\{i\}$ and rank-one product forbidden projectors
  \[
    \Pi_g=\bigotimes_{j\in I_g}P_{g,j},
    \qquad
    U_{g,j}=2P_{g,j}-I,
    \qquad g\in\{e,f\}.
  \]
  Then, as operators on incident variables $I_e\cup I_f$,
  \[
    \begin{aligned}
      \Pi_e+\Pi_f
      \succeq
      \frac{2-\sqrt{2+|\Tr(U_{e,i}U_{f,i})|}}8 I +\frac18\sum_{g\in\{e,f\}}\left(
        \sum_{j\in I_g}U_{g,j}
        +\sum_{\{j,j'\}\in{I_g\choose2}}U_{g,j}U_{g,j'}
      \right).
    \end{aligned}
  \]
\end{lemma}

\begin{proof} Similar to Fourier analysis for classical CSP in the $\{\pm 1\}$ basis,
  expanding the two product projectors gives
  \[
    \Pi_e+\Pi_f
    =\frac18\left(
      2I+\sum_{g\in\{e,f\}}\left(
        \sum_{j\in I_g}U_{g,j}
        +\sum_{\{j,j'\}\in{I_g\choose2}}U_{g,j}U_{g,j'}
        +\prod_{j\in I_g}U_{g,j}
      \right)
    \right).
  \]
  The shared-qubit reflections satisfy
  $\{U_{e,i},U_{f,i}\}=\Tr(U_{e,i}U_{f,i})I$, and the product of the remaining four reflections is again a reflection
  because they act on distinct qubits.  Consequently,
  \begin{align*}
    \left(\prod_{j\in I_e}U_{e,j}+\prod_{j\in I_f}U_{f,j}\right)^2
    &=2I+\Tr(U_{e,i}U_{f,i})
      \prod_{j\in I_e\setminus\{i\}}U_{e,j}
      \prod_{j\in I_f\setminus\{i\}}U_{f,j} \\
    &\preceq\left(2+|\Tr(U_{e,i}U_{f,i})|\right)I.
  \end{align*}
  Thus the sum of the two cubic terms is at least
  $-\sqrt{2+|\Tr(U_{e,i}U_{f,i})|}\,I$.  Substituting this bound into the expansion above proves
  the lemma.
\end{proof}

With this expansion, we are ready to introduce our key observation - this provides a non-trivial lower bound for the violation mass if the two constraints do not commute on the shared qubit! Before we elaborate on this, we first show that in a random instance, it suffices for us to consider these clean local structures.

\subsection{Finding Useful Bow-Ties}
\label{sec:useful-bowties}

\paragraph{Random model and parameters.}
Throughout this subsection, fix an integer $k\geq2$ and a density $\Delta>0$, both independent of $n$.
For every sufficiently large $n$, sample the instance $\cI$ by including each unordered scope
$e\in{[n]\choose k}$ independently with probability
\(
  p\defas\frac{\Delta n}{{n\choose k}}.
\)
For every included scope, independently choose a uniformly random ordering and a forbidden projector
$\Pi_e=\ket{\psi_e}\bra{\psi_e}$, where $\ket{\psi_e}$ is Haar-random in $(\C^2)^{\otimes k}$; let $\cE$ denote the
resulting set of ordered constraints.  Thus $|\cE|\sim\operatorname{Bin}({n\choose k},p)$ and
$\E|\cE|=\Delta n$.  All probabilities below are over the random scopes, their orderings, and their forbidden
projectors.
Fix constants $B>1$ and $c,\xi>0$, independent of $n$, satisfying
\[
  2\xi<c<\frac{2-\sqrt 2}{8}.
\]
The choices $B=100$, $c=1/20$, and $\xi=1/1000$ satisfy these requirements.

\begin{definition}[Useful] \label{def:useful}
  A pair of constraints $(e,f)$ of $\cI$ sharing exactly one vertex $i$ is a useful bow-tie if there are rank-one
  product projectors
  \[
    \widehat\Pi_g=\bigotimes_{j\in I_g}P_{g,j},
    \qquad
    U_{g,j}=2P_{g,j}-I,
    \qquad g\in\{e,f\},
  \]
  satisfying:
  \begin{enumerate}
    \item $\norm{\Pi_g-\widehat\Pi_g}_{\mathrm{op}}\leq\xi$ for $g\in\{e,f\}$.
    \item $\bigl(2-\sqrt{2+|\Tr(U_{e,i}U_{f,i})|}\bigr)/8\geq c$.
    \item Every vertex $v\in I_e\cup I_f$ belongs to at most $Bk\Delta$ constraints in
      $\cE\setminus\{e,f\}$.
  \end{enumerate}
  For each useful bow-tie, fix one such pair of product witnesses.
\end{definition}

\begin{lemma}[Useful bow-tie domination]
  \label{lem:useful-bowtie-domination}
  Let $(e,f)$ be a useful bow-tie in a quantum $3$-SAT instance.  Then, as operators on $I_e\cup I_f$,
  \begin{equation} \label{eq:useful-bowtie-domination}
    \Pi_e+\Pi_f
    \succeq
    (c-2\xi)I+\frac18\sum_{g\in\{e,f\}}\left(
      \sum_{j\in I_g}U_{g,j}
      +\sum_{\{j,j'\}\in{I_g\choose2}}U_{g,j}U_{g,j'}
    \right).
  \end{equation}
\end{lemma}

\begin{proof}
  The operator-norm bounds in \cref{def:useful} imply
  $\Pi_g\succeq\widehat\Pi_g-\xi I$ for $g\in\{e,f\}$.  Applying
  \cref{lem:exact-product-bowtie} to the two witnesses and using condition~2 of \cref{def:useful} therefore gives
  \[
    \Pi_e+\Pi_f
    \succeq\widehat\Pi_e+\widehat\Pi_f-2\xi I
    \succeq(c-2\xi)I+\frac18\sum_{g\in\{e,f\}}\left(
      \sum_{j\in I_g}U_{g,j}
      +\sum_{\{j,j'\}\in{I_g\choose2}}U_{g,j}U_{g,j'}
    \right).
  \]
\end{proof}

\begin{lemma} \label{lem:many-useful-bowties}
  With high probability over the random instance, $\cI$ contains $\Omega(\Delta^2 k^2 n)$ useful bow-ties.
\end{lemma}

\begin{proof}
  Let $X$ be the number of ordered pairs of constraints $(e,f)$ of $\cI$ that are useful bow-ties.
  We prove concentration of $X$ using the second moment method.

  We first compute the expectation
  \[ \E X = \sum_{(e,f)} \Pr[(e,f) \text{ useful}] \,, \]
  where the sum is over all potential bow-ties $(e,f)$.
  There are ${n \choose k} \sim n^k/k!$ choices for $e$, $k$ choices for the common vertex, and then ${n-k \choose k-1}
  \sim n^{k-1}/(k-1)!$ choices for the remaining vertices of $f$.
  Each constraint appears with probability $p$.
  Consider two product projectors whose reflections on the shared qubit are orthogonal in the Hilbert--Schmidt inner
  product.  For this pair, the left-hand side of condition 2 in \cref{def:useful} equals $(2-\sqrt2)/8>c$.
  By continuity, conditions 1 and 2 hold throughout an open neighborhood of this pair, and the full support of Haar
  measure implies that they therefore hold jointly with probability $r=r(k,c,\xi)>0$, independent of $n$ and
  $\Delta$.
  Standard joint factorial moment calculations and \cite[Theorem~6.10]{JLR00} show that the degrees of the $2k-1$
  vertices in $e \cup f$ (excluding the contribution of $e$ and $f$) converge to independent Poisson random variables
  $\text{Poi}(k\Delta)$ with parameter $k\Delta$.
  Therefore condition 3 in \cref{def:useful} holds with probability $(1+o(1))u^{2k-1}$, where
  $u = \Pr[\text{Poi}(k\Delta) \leq Bk\Delta] > 0$.
  Since $B>1$, a Poisson Chernoff bound also gives $u\geq1/2$ for every sufficiently large $\Delta$.
  Thus
  \[ \E X \sim {n \choose k} k {n-k \choose k-1} p^2ru^{2k-1} \sim \Delta^2 k^2ru^{2k-1} n \,. \]

  We now compute the second moment
  \[ \E[X^2] = \sum_{(e,f),(e',f')} \Pr[(e,f),(e',f') \text{ both useful}] = Y_{\text{disj}} + Y_{\text{non-disj}}
  \,, \]
  where the sum is over all potential bow-ties $(e,f),(e',f')$, and $Y_{\text{disj}}$ denotes the contribution by
  $(e,f)$ and $(e',f')$ on disjoint vertices.
  We first compute $Y_{\text{disj}}$.
  There are ${n \choose k}{n-k \choose k}$ choices for $e,e'$, $k$ choices for the common vertex of $(e,f)$, $k$ choices
  for the common vertices of $(e',f')$.
  There are ${n-2k \choose k-1}{n-3k+1 \choose k-1}$ choices for the remaining vertices of $f$ and $f'$.
  As before, conditions 1 and 2 in \cref{def:useful} hold for $(e,f)$ with probability $r$, and also for $(e',f')$
  with probability $r$.
  The degrees of the $4k-2$ vertices in $e \cup f \cup e'\cup f'$ (excluding $e,f,e',f'$) converge to independent
  Poissons with parameter $k\Delta$, and condition 3 in \cref{def:useful} holds with probability $(1+o(1))u^{4k-2}$.
  Altogether, such vertex-disjoint $(e,f)$ and $(e',f')$ contribute
  \begin{align*}
    Y_{\text{disj}}
    &\sim {n \choose k} {n-k \choose k} k^2 {n-2k \choose k-1}{n-3k+1 \choose k-1} p^4 r^2u^{4k-2} \\
    &\sim \Delta^4k^4r^2u^{4k-2}n^2 \,.
  \end{align*}
  We now loosely upper bound the contribution $Y_{\text{non-disj}}$ to $\E[X^2]$ by $(e,f),(e',f')$ that are not
  vertex-disjoint.
  If $(e,f),(e',f')$ share exactly $\ell$ vertices, the number $N_\ell$ of choices of $(e,f),(e',f')$ is at most
  $O_k(1)^\ell/(n-3k)^\ell$ times the number $N_0$ of choices for vertex-disjoint $(e,f),(e',f')$.
  When $(e,f),(e',f')$ share $\ell \geq 1$ vertices but not hyperedges, they contribute at most
  \[ N_\ell p^4 r^2 \lesssim_{\Delta,k} Y_{\text{disj}}/n \,. \]
  When $(e,f),(e',f')$ share exactly one hyperedge, $\ell \geq k$.
  They contribute at most
  \[ N_\ell p^3 \lesssim_{\Delta,k} Y_{\text{disj}}/n \,. \]
  When $(e,f),(e',f')$ share both hyperedges, then they contribute at most $2 \E X$ (the factor $2$ accounts for
  swapping the roles of $(e',f')$).
  Summing the contributions of these cases,
  \begin{align*}
    Y_{\text{non-disj}} \leq \sum_{1 \leq \ell \leq 2k-1} O_{\Delta,k}(Y_{\text{disj}}/n)
    + 2\E X \lesssim_{\Delta,k} Y_{\text{disj}}/n \,.
  \end{align*}
  Therefore $\E[X^2] \sim (\E X)^2$.
  Chebyshev inequality implies $X=\Omega(\Delta^2k^2n)$ with high probability.  Dividing by two to forget the
  ordering of $(e,f)$ proves the claim.
\end{proof}

\subsection{Strong Refutation Near Pairwise Uniformity}
\label{sec:near-uniform-refutation}

We wrap up this section by combining the above into a strong refutation for states that are close to pairwise uniform.
Recall $\tilde\cI$ and $\tilde\Psi$ from \cref{def:eps-net-instance}.
\renewcommand{\theHclaim}{\theclaim}
\begin{claim}[Strong refutation near pairwise uniformity]
  \label{claim:fractional-matching-refutation}
  For every sufficiently small constant $\gamma>0$ and every sufficiently large constant $\Delta$, there is a constant
  $\beta_{\mathrm{unif}}(\gamma)>0$ such that, with high probability over a random quantum $3$-SAT instance
  $\cI=([n],\cE)$ of expected density $\Delta$, for any collection of pairwise $\nu$-independent clause states
  $\{\rho_{\tilde\psi}\}_{\tilde\psi \in \tilde\Psi}$ with $\norm{\nu-I/2}_1\leq\gamma$,
  \[ \val(\tilde\cI,\{\rho_{\tilde\psi}\})\le  1-\beta_{\mathrm{unif}}(\gamma) + 2\eps \,. \]
\end{claim}

\begin{proof}
  For sufficiently small $\gamma$,
  \[
    c-2\xi-\frac34(\gamma+\gamma^2)>0;
  \]
  for the numerical choices in \cref{sec:useful-bowties}, $\gamma=1/100$ suffices.  By
  \cref{lem:many-useful-bowties} and a Chernoff bound, with high probability the instance contains a collection
  $\cB$ of $\Omega(\Delta^2k^2n)$ useful bow-ties and has $m=|\cE|=\Theta(\Delta n)$.  Fix an instance in this event.

  For any $\{e,f\}\in\cB$, applying \eqref{eq:useful-bowtie-domination} to the $3$-qubit states $\rho_e,\rho_f$ and
  using pairwise $\nu$-independence gives
  \[
    \Tr(\Pi_e\rho_e)+\Tr(\Pi_f\rho_f)
    \geq c-2\xi-\frac34(\gamma+\gamma^2).
  \]
  The degree condition in \cref{def:useful} implies that each clause belongs to at most $k(\Delta kB+2)$ useful
  bow-ties.  Nonnegativity of the clause violations therefore gives
  \begin{align*}
    1-\val(\cI,\{\rho_{\tilde\psi}\})
    =\frac1m\sum_{e\in\cE}\Tr(\Pi_e\rho_e)
    &\geq\frac{1}{m k(\Delta kB+2)}
    \sum_{\{e,f\}\in\cB}\left(\Tr(\Pi_e\rho_e)+\Tr(\Pi_f\rho_f)\right) \\
    &\geq
    \left(c-2\xi-\frac34(\gamma+\gamma^2)\right)
    \frac{|\cB|}{m k(\Delta kB+2)}.
  \end{align*}
  The final ratio is $\Omega(1)$ by \cref{lem:many-useful-bowties}. Taking $\beta_{\mathrm{unif}}(\gamma)$ to be the
  corresponding constant proves $1-\val(\cI,\{\rho_{\tilde\psi}\}) \geq \Omega(1) =: \beta_{\mathrm{unif}}(\gamma)$.
  The bound on $\val(\tilde\cI, \{\rho_{\tilde\psi}\})$ follows by \cref{lem:outer-product-diff-norm}.
\end{proof}

\section{Conditioning} \label{sec:conditioning}

\ConditionalIndependence

\begin{proof}
  Since $\cE$ is $(\eta/(2q),q,2,2)$-concentrated, for every distinct $i,j \in [k]$,
  \[ \E_{T,b} \norm{\cE(\rho|_b)_{ij} - \E_{\text{distinct } u,v \in [n]} (\rho|_b)_{uv}}_F
  \leq \frac{\eta}{2q} \,. \]
  Let $R = R^{T,b}$ be the operator in the above Frobenius norm.
  The Frobenius norm of a Hermitian operator $R$ on two $q$-ary qudits upper bounds its Schatten $1$-norm by
  \[ \norm{R}_1 \leq {\textstyle\sum_r \abs{\lambda_r(R)}} \leq
  \sqrt{\textstyle \sum_r \lambda_r(R)^2}\sqrt{\textstyle \sum_r 1^2} \leq q \norm{R}_F\,, \]
  where $\lambda_r(R)$ are the eigenvalues of $R$.
  Let $(\mu|_b)_S \defas \Lambda^{\otimes S}((\rho|_b)_S)$.
  \cref{eq:distortion} implies
  \[ \norm{(\rho|_b)_{ij} - (\rho|_b)_i \otimes (\rho|_b)_j}_1 \leq 18q \cdot
  \norm{(\mu|_b)_{ij} - (\mu|_b)_i \otimes (\mu|_b)_j}_1 \,. \]
  The result now follows from \cref{thm:cond-indep-complete} below (with its $\eps = \eta/(36q)$).
\end{proof}

The next theorem is inspired by \cite{RaghavendraT12} and \cite{BH13} and is a quantum variant of
\cite[Theorem~5.6]{BRS11}.
\cite[Theorem~14]{BH13} is another quantum variant of \cite[Theorem~5.6]{BRS11} that does not work for our application,
because we need $\E_b (\rho|_b)_S \approx \rho_S$ on a $3$-QSAT constraint with scope $S$.
For this reason, our conditioning differs from the conditioning of local classical distributions of
\cite[Theorem~14]{BH13}.
Unlike the classical case, note that $\E_b (\rho|_b)_S$ is only approximately but not exactly $\rho_S$, due to the
measurement for the conditioning.

\begin{theorem} \label{thm:cond-indep-complete}
  Given the complete graph $G = (V, E)$ on $n$ vertices, there is $\kappa = O(\eps^{-2}\log q)$ such that, for every
  $n \geq 4\kappa/\eps$, every degree $2(\kappa+k)$ $q$-ary pseudo-state $\rho$ satisfies
  \begin{equation} \label{eq:cond-indep}
    \E_{T \sim \cT} \E_{b \sim \mu_T} \E_{ij \sim E}
    \norm{(\mu|_b)_{ij} - (\mu|_b)_i \otimes (\mu|_b)_j}_1 \leq \eps \,.
  \end{equation}
  where $(\mu|_b)_S \defas \Lambda^{\otimes S}((\rho|_b)_S)$.
  Further, for every subset $S$ of size $k \geq 2$,
  \begin{equation} \label{eq:approx-total-prob}
    \norm{\E_{T,\, b} (\rho|_b)_S - \rho_S}_1 \leq \kappa k/n\,.
  \end{equation}
\end{theorem}

\begin{proof}
  Our proof is a variant of \cite[Theorem~14]{BH13}.
  Instead of their classical conditioning, we consider the post-measurement pseudo-state $\rho|_b$, in order
  to approximately recover the law of total probability.
  We also use the (classical) mutual-information analysis from \cite{RaghavendraT12}.
  Let $d \defas \kappa+k$.
  The pseudo-state $\rho$ has degree $2d$.
  Therfore $\rho|_b$ is a pseudo-state of degree $2k$ (\cref{lem:conditioning}).

  \begin{claim}
    The approximate law of total probability \cref{eq:approx-total-prob} holds.
  \end{claim}

  \begin{proof}
    Given any subset $T \subseteq [n]$ of size at most $d$, the channel $\Lambda^{\otimes T}$ has POVM $\{\cE_b\}$
    with Kraus operators $\sqrt{\cE_b}$, where $\cE_b = I_{[n] \setminus T} \otimes \bigotimes_{t \in T} \cE_{t,b(t)}$,
    and $\{\cE_{t,x}\}$ is the POVM for the single-qudit channel $\Lambda$ for $t \in T$.
    For any operator $Y$ of degree $\leq 2(d-s)$,
    \[
      \E_{b \sim \mu_T} \pE_{\rho|_b} Y
      = \sum_{b \in \supp(\mu_T)} \pE_\rho [\cE_b]
      \frac{\pE_\rho [\sqrt{\cE_b} Y \sqrt{\cE_b}]}{\pE_\rho [\cE_b]}
      = \pE_\rho \sum_{b \in \supp(\mu_T)} \sqrt{\cE_b} Y \sqrt{\cE_b} \,.
    \]
    The final sum remains unchanged if $b$ runs over all of $[q']^T$, because for $b \in [q']^T \setminus \supp(\mu_T)$,
    we have $\pE_\rho \cE_b = \mu_T(b) = 0$, so
    \[ \Abs{\textstyle \pE_\rho \sqrt{\cE_b}\, Y\, \sqrt{\cE_b}}
    \leq \sqrt{\textstyle \pE_\rho \cE_b} \sqrt{\textstyle \pE_\rho \sqrt{\cE_b} Y^\dagger Y\sqrt{\cE_b}} = 0\,, \]
    by pseudo-expectation Cauchy--Schwarz (\cref{lem:pseudo-cs}).
    Therefore
    \[ \E_{b \sim \mu_T} \pE_{\rho|_b} Y = \pE_\rho \sum_{b \in [q']^T} \sqrt{\cE_b} Y \sqrt{\cE_b}
    = \pE_\rho \Lambda^{\otimes T} (Y)\,. \]
    In particular, if $Y = I_T \otimes X$ for some operator $X$ on $\overline T = [n] \setminus T$ of degree
    $\leq 2(d-s)$, then $\E_{b \sim \mu_T} \pE_{\rho|_b} I_T \otimes X = \pE_\rho I_T \otimes X$, because
    $\Lambda^{\otimes T}$ is unital.
    Thus $\E_{b \sim \mu_T} (\rho|_b)_{\overline T} = \rho_{\overline T}$ as reduced pseudo-states, and
    $\E_{b \sim \mu_T} (\rho|_b)_S = \rho_S$ for any $S \subseteq \overline T$.
    Then
    \[ \norm{\E_{T,b}(\rho|_b)_S-\rho_S}_1 \leq \Pr_{T\sim\cT}[T\cap S\neq\emptyset]\cdot 2\,, \]
    since $\E_b(\rho|_b)_S$ and $\rho_S$ are equal unless $T$ meets $S$, and otherwise both are normalized $k$-qudit
    density matrices whose trace-norm difference is at most $2$.
    As $s$ is uniform on $\{0,\dots,\kappa\}$, $\E[s]=\kappa/2$, and each of the $s$ entries hits $S$ with probability
    $k/n$, so $\Pr_{T\sim\cT}[T\cap S\neq\emptyset]\le\frac{k}{n}\E[s]=\frac{\kappa k}{2n}$, giving the bound
    $\frac{\kappa k}{n}$ in \cref{eq:approx-total-prob}.
  \end{proof}

  To bound the average $1$-norm distance in \cref{eq:cond-indep}, we first show the distributions
  in \cref{eq:cond-indep} from quantum conditioning $(\mu|_b)_S = \Lambda^{\otimes S}((\rho|_b)_S)$
  coincide with those from classical conditioning $\mu(\cdot|b)_S = (\Lambda^{\otimes S \cup T}(\rho_{S \cup T})|_b)_S$,
  provided $S$ is disjoint from $T$ and $\abs{S \cup T} \leq d$.

  \begin{claim} \label{claim:quantum-classical-conditioning}
    For any disjoint subsets $S$ and $T$ of $[n]$ satisfying $\abs{S \cup T} \leq d$, any $b \in
    \supp(\Lambda^{\otimes T}(\rho_T))$, we have
    $\Lambda^{\otimes S}((\rho|_b)_S) = (\Lambda^{\otimes S \cup T}(\rho_{S \cup T})|_b)_S$.
  \end{claim}

  \begin{proof}
    $b \in \supp(\Lambda^{\otimes T}(\rho_T))$ is equivalent to $\pE_\rho[\cE_b] > 0$.
    For the distribution $\mu_{S \cup T} = \Lambda^{\otimes S\cup T}(\rho_{S \cup T})$, classically conditioning on
    $x_T = b$ yields a distribution on the free sites $S$
    \[ \Lambda^{\otimes S\cup T}(\rho_{S \cup T})\big|_{b}(c)
      = \frac{\mu_{S \cup T}(b,c)}{\mu_T(b)}
      \overset{\eqref{eq:quantum-classical-channel}}= \frac{\pE_\rho[\cE_b\,\cE_c]}{\pE_\rho[\cE_b]}
    \qquad \text{for } c \in [q']^S \,. \]
    On the other hand, the quantum-conditioned distribution is
    \[
      \Lambda^{\otimes S}\Paren{(\rho|_b)_S}(c) \overset{\eqref{eq:quantum-classical-channel}}= \pE_{\rho|_b} [\cE_c]
      \overset{\eqref{eq:conditioning}}= \frac{\pE_\rho[\sqrt{\cE_b}\,\cE_c\,\sqrt{\cE_b}]}{\pE_\rho[\cE_b]}
      \overset{(*)}= \frac{\pE_\rho[\cE_b\,\cE_c]}{\pE_\rho[\cE_b]}
      = \Lambda^{\otimes S\cup T}(\rho_{S \cup T})\big|_{b}(c)
    \]
    for $c \in [q']^S$, where $(*)$ uses the commutivity of $\sqrt{\cE_b}$ and $\cE_c$ due to disjointness
    of $S$ and $T$.
  \end{proof}

  \cref{eq:cond-indep} now follows from
  \begin{equation} \label{eq:cond-indep-classical}
    \E_{T,\, b} \E_{ij \sim E} \norm{\mu(\cdot|b)_{ij} - \mu(\cdot|b)_i \otimes \mu(\cdot|b)_j}_1 \leq \eps/2 \,.
  \end{equation}
  Indeed, \cref{claim:quantum-classical-conditioning} implies $\mu(\cdot|b)_{ij} = (\mu|_b)_{ij}$ whenever $\{i, j\}$
  is disjoint from $T$, and similarly $\mu(\cdot|b)_i = (\mu|b)_i$.
  Also, $\Pr[\{i ,j\} \cap T \neq \emptyset] \leq \kappa/n$, and the $1$-norm difference in
  \cref{eq:cond-indep} is at most $2$ for each $(T,b,i,j)$.

  Pinsker's and Cauchy--Schwarz inequalities imply the LHS of \cref{eq:cond-indep-classical} is at most
  \begin{align*}
    &\sqrt{2\ln 2\,\E_T \E_{ij \sim E} \E_{b\sim\mu_T}
    D(\mu(\cdot|_b)_{ij} \;\|\; \mu(\cdot|_b)_i \otimes \mu(\cdot|_b)_j)} \\
    =\; &\sqrt{2\ln 2\cdot \E_T \E_{ij \sim E} I(x_i;x_j|x_T)} \,,
  \end{align*}
  where $(x_T,x_i,x_j)$ jointly distributes as the classical distribution $\mu_{T \cup \set{i,j}}$.
  When $i \in T$ or $j \in T$, $I(x_i;x_j|x_T) = 0$.
  Therefore the only contribution to the conditional mutual-information comes from $i,j \in \overline T \defas [n]
  \setminus T$.
  Thus
  \begin{equation} \label{eq:mutual-info-edge}
    \E_{ij \sim E} I(x_i;x_j|x_T) = \sum_{\substack{i, j \in \overline T \\ i \neq j}}
    \frac{I(x_i;x_j|x_T)}{n(n-1)} = \frac{Q - D}{n(n-1)} \,,
  \end{equation}
  where
  \[ Q = \sum_{i,j \in \overline T} I(x_i;x_j|x_T)\,, \qquad D = \sum_{i \in \overline T} H(x_i|x_T) \,. \]
  Since $D$ is the diagonal terms of $Q$ and conditional mutual information is non-negative, $Q \geq D$.
  Recall $|T| = s$.
  Rearranging, \cref{eq:mutual-info-edge} is at most
  \[ \frac Q{n(n-s)} = \sum_{i \in [n]} \sum_{j \in \overline T} \frac{I(x_i;x_j|x_T)}{n(n-s)}
  = \E_{i \in [n]} \E_{j \in \overline T} I(x_i;x_j|x_T) \,. \]

  We bound the conditional mutual-information following \cite[Lemma~4.5]{RaghavendraT12}, whose analysis applies
  because the local distributions $\set{\mu_T}$ are consistent under projection, forming a level-$d$ Sherali--Adams
  solution (\cref{lem:partial-trace-channel}).
  Since $I(x_i;x_j|x_T) = H(x_i|x_T) - H(x_i|x_T,x_j)$,
  \[ \E_{|T|=s} \E_{i \in [n]} \E_{j \in \overline T} I(x_i;x_j|x_T)
  = \E_{|T|=s} \E_{i \in [n]} H(x_i|x_T) - \E_{|S|=s+1} \E_{i \in [n]} H(x_i|x_S) \,. \]
  Averaging over all $s \in \set{0, \dots, \kappa}$, the RHS telescopes, so
  \[ \E_{T \in \cT} \E_{i \in [n]} \E_{j \in \overline T} I(x_i;x_j|x_T) = \frac1{\kappa+1}
  \Paren{\E_{i \in [n]} H(x_i) - \E_{|S| = \kappa+1} H(x_i|x_S)} \leq \frac{\E_i H(x_i)}{\kappa+1} \,, \]
  because conditional entropy is non-negative.
  The RHS is $O(\log q)/(\kappa+1)$, as $H(x_i) \leq \log q' = O(\log q)$.
  This proves \cref{eq:cond-indep-classical} whenever $\kappa \geq \Omega(\eps^{-2}\log q)$.
\end{proof}

\section{Concentration} \label{sec:concentration}

\Concentration*

\begin{proof}
  As mentioned in \cref{rem:pseudo-moment-matrix}, the pseudo-expectation $\pE_\rho$ of a pseudo-state $\rho$
  corresponds to a positive semidefinite matrix $M$.
  Let $d = t/2$.
  If $\rho$ has degree $t = 2d$, then $M$ is $\cW^n_d$-by-$\cW^n_d$, where $\cW^n_d$ is the set of $q$-ary $n$-qudit
  Heisenberg--Weyl operators of weight at most $d$.

  Fix $W \in {[k] \choose t}$ and let $\chi \defas \E_{e \in \cE} (\rho_e)_W - \E_{e \in V^{\underline k}} (\rho_e)_W$.
  Our goal is to bound $\norm{\chi}_F^2$.
  Expand $\norm{\chi}_F^2 = \frac 1{q^t} \sum_{\gamma \in \cW^t} \abs{\Tr \chi \gamma}^2$, because Heisenberg--Weyl
  operators $\cW^t$ on $t$ qudits are orthogonal and have Frobenius norm $\sqrt{q^t}$.
  We will first bound $\abs{\Tr \chi \gamma}$ for $\gamma \in \cW^t$.
  Given $e \in [n]^{\underline k}$, let $e_W \in [n]^{\underline t}$ be the $t$-tuple such that $(e_W)_i = e(W_i)$ for
  $i \in [t]$, where $W_1 < \dots < W_t$ are the elements of $W$.
  Note that $\Tr[(\rho_e)_W \gamma] = \Tr[\rho (e_W\circ \gamma)]$, where $e_W\circ \gamma$ is the Heisenberg--Weyl
  operator whose $(e_W)_i$-component is $\gamma_i$ for $i \in [t]$ (and is the identity outside of $e_W$).
  Thus
  \[ \Tr[\chi \gamma] = \E_{e \in \cE} \Tr[\rho (e_W\circ\gamma)]
  - \E_{e \in V^{\underline k}} \Tr[\rho (e_W\circ\gamma)] = \Tr[(L^{\gamma,W})^\dagger M] \,, \]
  for some $\cW^n_d$-by-$\cW^n_d$ matrix $L^{\gamma,W}$.
  More precisely, $L^{\gamma,W}$ has entries
  \begin{equation} \label{eq:deviation-matrix}
    \begin{split}
      L^{\gamma,W}_{\sigma\tau} \defas
      \quad & \Pr_{e \in \cE} [e_{W^\text{left}} \circ \gamma^\text{left} = \sigma^\dagger,
      e_{W^\text{right}} \circ \gamma^\text{right} = \tau] \\
      - & \Pr_{e \in V^{\underline k}} [e_{W^\text{left}} \circ \gamma^\text{left} = \sigma^\dagger,
      e_{W^\text{right}} \circ \gamma^\text{right} = \tau] \,,
    \end{split}
  \end{equation}
  for $\sigma,\tau \in \cW^n_d$, where $W^\text{left} = (W_1, \dots, W_d)$ and $W^\text{right} =
  (W_{d+1}, \dots, W_t)$ are the first and second halves of $W$, and similarly $\gamma = \gamma^\text{left} \otimes
  \gamma^\text{right}$ is decomposed into Heisenberg--Weyl operators on the first and second $d$ qudits.

  For any $\gamma \in \cW^t$, \cref{lem:deviation-bound} implies $\abs{\Tr \chi\gamma} \leq \eta/\sqrt{q^t}$
  except with probability $o_n(1)$.
  Then $(\eta,q,t,t)$-concentration follows by a union bound over all $\gamma$ and $W$.
\end{proof}

\begin{lemma} \label{lem:deviation-bound}
  Fix even $t \leq k$, $q$-ary $\gamma \in \cW^t$, $W \in {[k] \choose t}$.
  For any $\eta > 0$, there is $C = C(\eta,q,t)$ such that for $m \geq Cn^{t/2}$, except with probability $o_n(1)$ over
  a directed $k$-uniform hypergraph $\cH = ([n],\cE)$ with $m$ expected hyperedges,
  the matrix $L^{\gamma,W}$ from \cref{eq:deviation-matrix} satisfies
  \begin{equation} \label{eq:deviation-bound}
    \sup_M \Abs{\Tr[(L^{\gamma,W})^\dagger M]} \leq \eta \,,
  \end{equation}
  where the supremum is over all matrices $M$ of a degree-$t$ $q$-ary $n$-qudit pseudo state $\rho$.
\end{lemma}

\begin{proof}
  Since $M$ is positive semidefinite and with diagonal entries $1$ (because matrices in $\cW^n_d$ are unitary), $M$ is
  the Gram matrix of some unit vectors $\set{v_\sigma}$.
  Therefore the supremum in \cref{eq:deviation-bound} is at most
  \begin{equation} \label{eq:grothendieck-opt}
    \sup \Abs{\sum_{\sigma, \tau} L^{\gamma,W}_{\sigma\tau} u_\sigma^\dagger v_\tau}
  \end{equation}
  over unit vectors $u_\sigma, v_\tau \in \C^D$ of arbitrary dimension $D$.
  The complex Grothendieck inequality implies that the latter supremum
  is at most the complex Grothendieck constant $K^\C_G$ times \cref{eq:grothendieck-opt} over unit magnitude complex
  numbers $u_\sigma, v_\tau \in \C$ (e.g.~\cite{friedland2018elementary}).

  Let $\cN \subset \S$ be a finite $\eps$-net of the complex unit sphere $\S = \set{ z \in \C | \abs{z} = 1}$.
  For any $u \in \S^{\cW^n_d}$, round $u$ coordinate-wise to obtain $\tilde u$, so that $\tilde u_\sigma$ is the closest
  point in $\cN$ to $u_\sigma \in \S$ for $\sigma \in \cW^n_d$.
  Likewise round $v \in \S^{\cW^n_d}$ coordinate-wise to obtain $\tilde v$.
  Then
  \[ \Abs{\sum_{\sigma, \tau} L^{\gamma,W}_{\sigma\tau} u_\sigma^\dagger v_\tau
  - \sum_{\sigma, \tau} L^{\gamma,W}_{\sigma\tau} \tilde u_\sigma^\dagger \tilde v_\tau} \leq 4\eps\,, \]
  because $\sum_{\sigma,\tau}\abs{L^{\gamma,W}_{\sigma\tau}} \leq 2$ by \cref{eq:deviation-matrix} and
  $\abs{u_\sigma^\dagger v_\tau - \tilde u_\sigma^\dagger \tilde v_\tau} \leq 2\eps$ by triangle inequality.
  That means to bound \cref{eq:grothendieck-opt}, we can first bound the supremum over a finite set of vectors $u,v \in
  \cN^{\cW^n_d}$ by $\eta/2K^\C_G$, and then the supremum over $u,v \in \S^{\cW^n_d}$ will be at most $\eta/K^\C_G$,
  whenever $\eps = \eta/8K^\C_G$.

  Fix any $u,v \in \cN^{\cW^n_d}$.
  Recall that each undirected $\hat e \in {[n] \choose k}$ appears in $\cH$ independently with probability $p$, and each
  realized $\hat e$ is uniformly ordered as an ordered $k$-tuple $e \in [n]^{\underline k}$.
  For each $\hat e \in {[n] \choose k}$, define the random variable
  \[ X_{\hat e} = \Iv{\hat e \in \cH} u^\dagger_{\sigma(e)} v_{\tau(e)} \,, \]
  where $\sigma(e) \defas e_{W^\text{left}} \circ (\gamma^\text{left})^\dagger$ and $\tau(e) \defas
  e_{W^\text{right}} \circ \gamma^\text{right}$.
  Then random sum
  \[ S^{uv} \defas \sum_{\hat e \in {[n] \choose k}} X_{\hat e} \]
  has expectation
  \[ \E_\cH S^{uv} = m \E_{e \in V^{\underline k}} u^\dagger_{\sigma(e)} v_{\tau(e)} \,, \]
  and we have used $m = p {n \choose k}$.
  Then the mean-$0$ random variables $Y_{\hat e} \defas X_{\hat e} - \E X_{\hat e}$ satisfy
  \begin{align*}
    \abs{Y_{\hat e}} &\leq \abs{X_{\hat e}} + \abs{\E X_{\hat e}} \leq 2 \sup_e \abs{u_{\sigma(e)}^\dagger
    v_{\tau(e)}} \leq 2 \\
    \E [Y_{\hat e}^2] &= \Var[X_{\hat e}] \leq \E[X_{\hat e}^2] \leq p
    \sup_e \abs{u_{\sigma(e)}^\dagger v_{\tau(e)}}^2 \leq p \,.
  \end{align*}
  Bernstein inequality (\cref{lem:bernstein}) applied to the real and imaginary parts implies
  \[ \Pr_\cH \Brac{\abs{S^{uv} - \textstyle \E_\cH S^{uv}} \geq \eta m/2K^\C_G} \leq \exp(-\Omega(\eta^2 m)) \,. \]
  Provided $m \geq C(\eta,q,t) n^d$, taking a union bound over all $\abs{\cN^{\cW^n_d}}^2 =
  \exp(O_\eta(n)^d)$ choices of $(u,v)$, whp $\abs{S^{uv} - \E_\cH S^{uv}} < \eta m/2K^\C_G$ simultaneously for all
  $u,v \in \cN^{\cW^n_d}$.
  This implies \cref{eq:grothendieck-opt} is at most $\eta/2K^\C_G$ over $u,v \in \cN^{\cW^n_d}$.
\end{proof}

\begin{lemma}[Bernstein Inequality] \label{lem:bernstein}
  Let $Y_1, \dots, Y_\ell$ be independent $0$-mean random variables such that $\abs{Y_e} \leq B$ for $e \in [\ell]$.
  Then for any $a > 0$,
  \[ \Pr\Brac{\sum_{e \in [\ell]} Y_e > a} \leq
  \exp\Paren{\frac{-\frac12 a^2}{\sum_{e \in [\ell]} \E[Y_e^2] + \frac 13 Ba}}. \]
\end{lemma}

\section*{Acknowledgments}
We thank Tomasso d'Orsi for discussion on an initial stage of this research. This question was first communicated to J.X. by Xinyu Wu in 2020. J.X. also thanks Pravesh K. Kothari, Tselil Schramm, Aaron Potechin, and Neng Huang for many helpful discussions at various points over the years. 
%

\paragraph{AI Disclosure:} The authors formulated the two key ingredients of the refutation algorithm: the reduction to pairwise-independent
states through the classical analogue, and the strong refutation of pairwise-uniform states using bow-tie structures.
AI assistance was used to work out the technical details. We used AI models GPT Pro, Laguna S 2.1, Ornith 1.5 35B A3B,
Qwen 3.8 27B to assist with

\begin{itemize}
  \item The quantum channels $\Phi^{\otimes k},\Psi^{\otimes k}$ and their analysis in the proof
    \cref{thm:almost-twise-indep} were found after discussing with Ornith 1.5 and Qwen 3.8 over multiple days, together
    with our simplification.
  \item It is folklore that a quantum $k$-SAT clause whose forbidden vector is a product state admits a satisfying
    pairwise-uniform state, by extending the classical construction for a forbidden computational-basis assignment.
    Through conversations with GPT Pro 5.5, we learned that, for $k=3$, the same conclusion holds for an arbitrary,
    possibly entangled, forbidden vector; see Appendix~\ref{sec:pairwise-uniform-state-proof}.  Although this observation
    is not used in our refutation algorithm, we include an independently verified proof because it may be of independent
    interest.
  \item Our quantum analogue of \cite[Theorem~5.6]{BRS11} (\cref{thm:cond-indep-complete}) was found by
    suggesting our high-level proof approach to Ornith 1.5 and discussing for multiple days, together with our
    simplification.
\end{itemize}
The authors verified the correctness and originality of all content including references.

\clearpage
\newpage
\appendix
\clearpage
\newpage

\clearpage
\newpage
\bibliographystyle{alpha}
\bibliography{main}
\clearpage
\newpage
\appendix
\section{Deferred Lemmas}
\label{sec:easy-results}

\begin{lemma} \label{lem:outer-product-diff-norm}
  For any pure states $\psi$ and $\phi$,
  \[ \norm{\ket\psi\bra\psi - \ket\phi\bra\phi}_\infty \leq 2\norm{\psi-\phi}\,. \]
\end{lemma}

\begin{proof}
  For any pure state $\sigma$,
  \begin{equation} \label{eq:pure-state-test}
    \begin{split}
      \abs{\Tr(\ket\sigma\bra\sigma(\ket\psi\bra\psi - \ket\phi\bra\phi))}
      &= \abs{\abs{\braket{\sigma|\psi}}^2 - \abs{\braket{\sigma|\phi}}^2} \\
      &= (\abs{\braket{\sigma|\psi}} + \abs{\braket{\sigma|\phi}}) \cdot \abs{\abs{\braket{\sigma|\psi}} -
      \abs{\braket{\sigma|\phi}}} \,.
    \end{split}
  \end{equation}
  We bound the first factor
  \[ \abs{\braket{\sigma|\psi}} + \abs{\braket{\sigma|\phi}} \leq \norm{\sigma}\norm{\psi} +
  \norm{\sigma}\norm{\phi} = 2 \]
  by Cauchy--Schwarz.
  We also bound the second factor
  \[ \abs{\abs{\braket{\sigma|\psi}} - \abs{\braket{\sigma|\phi}}} \leq \abs{\braket{\sigma|\psi} -
  \braket{\sigma|\phi}} \leq \norm{\sigma} \norm{\psi-\phi} \]
  by triangle inequality and then Cauchy--Schwarz.
  Plugging the previous two equations into \cref{eq:pure-state-test},
  \[ \abs{\Tr(\ket\sigma\bra\sigma(\ket\psi\bra\psi - \ket\phi\bra\phi))} \leq 2\norm{\psi-\phi} \,. \]
  The supremum over $\sigma$ of the left-hand-side equals $\norm{\ket\psi\bra\psi - \ket\phi\bra\phi}_\infty$ by the
  spectral theorem.
\end{proof}

\section{Pseudo-Expectation}

\begin{lemma} \label{lem:conditioning}
  Given any degree-$2(d+s)$ pseudo-expectation $\pE_\rho$, the conditional pseudo-expectation $\pE_{\rho|_x}$ in
  \cref{def:conditioning} is a pseudo-expectation of degree $2d$, whenever $\pE_\rho[\cE_x] > 0$.
\end{lemma}

\begin{proof}
  $\pE_{\rho|_x}[I] = \pE_\rho[\sqrt{\cE_x}\sqrt{\cE_x}]/\pE_\rho[\cE_x] = 1$.
  Given any $H \in \cL^n_d$, $\sqrt{\cE_x}H^\dagger H\sqrt{\cE_x} = G^\dagger G$ with $G = H\sqrt{\cE_x}
  \in \cL^n_{d+s}$, so $\pE_{\rho|_x}[H^\dagger H] = \pE_\rho[G^\dagger G]/\pE_\rho [\cE_x] \geq 0$.
\end{proof}

\begin{lemma} \label{lem:reduced-pseudo-expectation}
  Given any degree-$2d$ $n$-qudit pseudo-expectation $\pE_\rho$, any $T \subseteq [n]$, the reduced pseudo-expectation
  $\pE_{\rho_T}$ is a pseudo-expectation of degree $2d$.
\end{lemma}

\begin{proof}
  For every operator $H \in \cL^T_d$, $G \defas H \otimes I_{[n] \setminus T} \in
  \cL^n_d$ and $(H^\dagger H) \otimes I_{[n] \setminus T} = G^\dagger G$, so $\pE_{\rho_T}$ inherits positivity from
  $\pE_\rho$.
  Finally, $\pE_{\rho_T} [I_T] = \pE_\rho[I_T \otimes I_{[n] \setminus T}] = 1$ by normalization of $\pE_\rho$.
\end{proof}

\begin{lemma} \label{lem:genuine-state}
  Given an $n$-qudit pseudo-expectation operator $\pE_\rho$ of degree $2n$, there is a unique genuine quantum state
  $\hat\rho$ such that $\pE_\rho [H] = \Tr[\hat\rho H]$ for every operator $H$.
\end{lemma}

\begin{proof}
  Note that $\pE_\rho P \geq 0$ for every positive semi-definite operator $P$ in $\cL^n_n$, because $P =
  \sqrt{P}^\dagger\sqrt{P}$ with $\sqrt{P} \in \cL^n_n$, and $\pE_\rho$ is positive.
  Then $\pE_\rho H$ is real for every Hermitian $H$, because $H = P - N$ for some positive-semidefinite $P$ and $N$.
  Therefore $\pE_\rho [X]^\dagger = \pE_\rho [X^\dagger]$ for every operator $X$, by expanding $X = A + iB$ for
  Hermitian $A$ and $B$ and linearity of $\pE_\rho$.

  Define
  \[ \hat\rho \defas \frac1{q^n} \sum_{G \in \cW^n} \pE_\rho [G^\dagger] G\,. \]
  Then
  \[ \hat\rho^\dagger = \frac1{q^n} \sum_{G \in \cW^n} \pE_\rho [G^\dagger]^\dagger G^\dagger
  = \frac1{q^n} \sum_{G \in \cW^n} \pE_\rho [G] G^\dagger = \hat\rho \]
  where the middle equality uses $\pE_\rho[G^\dagger]^\dagger = \pE_\rho G$ and $G^\dagger \in \cW^n$ for $G \in \cW^n$.
  For any $H \in \cW^n$,
  \[ \Tr[\hat\rho H] = \frac1{q^n} \sum_{G \in \cW^n} \pE_\rho[G^\dagger] \Tr[GH] = \pE_\rho[H]\,, \]
  because $\cW^n$ is an orthogonal basis of $\cL^n$, each of Frobenius norm $\sqrt{q^n}$.
  In particular, $\Tr[\hat\rho] = \pE_\rho [I] = 1$, since $\pE_\rho$ is normalized.
  Also, $\Tr[\hat\rho H] = \pE_\rho[H]$ for every operator $H$ by linearity.
  Finally, for any operator $H$ (necessarily of degree at most $n$),
  $\Tr[\hat\rho H^\dagger H] = \pE_\rho [H^\dagger H] \geq 0$ since $\pE_\rho$ is positive,
  implying $\hat\rho$ is positive semidefinite.
  Uniqueness follows from $\norm{\hat\rho - \hat\sigma}_F^2 = \frac 1{q^n} \sum_{G \in \cW^n}
  \abs{\Tr [\hat\rho G] - \Tr[\hat\sigma G]}^2$.
\end{proof}

\begin{lemma} \label{lem:partial-trace-channel}
  For any operator $\rho$ on a set $S$ of qudits, any single-qudit trace-preserving linear map $\Lambda$, any
  $R \subseteq S$, we have $\Tr_R \Lambda^{\otimes S}(\rho) = \Lambda^{\otimes (S \setminus R)}(\Tr_R \rho)$.
\end{lemma}

\begin{proof}
  Let $T = S \setminus R$.
  Expand $\rho = \sum_i A_i \otimes B_i$, where $A_i$ acts on $T$ and $B_i$ acts on $R$.
  \begin{align*}
    \Tr_R \Lambda^{\otimes S}(A_i \otimes B_i)
    &= \Tr_R (\Lambda^{\otimes T}(A_i) \otimes \Lambda^{\otimes R}(B_i)) \\
    &= \Lambda^{\otimes T}(A_i) \Tr_R \Lambda^{\otimes R} (B_i) \\
    &\overset{(*)}= \Lambda^{\otimes T}(A_i) \Tr B_i = \Lambda^{\otimes T}(A_i \Tr B_i) \\
    &= \Lambda^{\otimes T}(\Tr_R A_i \otimes B_i) \,,
  \end{align*}
  where $(*)$ is the trace-preservation of $\Lambda^{\otimes R}$.
  The result follows by summing over all $i$ and the linearity of the partial trace and $\Lambda$.
\end{proof}

\begin{lemma}[Pseudo-expectation Cauchy--Schwarz] \label{lem:pseudo-cs}
  For any $P, Q \in \cL^n_d$, any degree-$2d$ pseudo-expectation $\pE_\rho$,
  \[ \Abs{\pE_\rho P^\dagger Q} \leq \sqrt{\pE_\rho P^\dagger P} \sqrt{\pE_\rho Q^\dagger Q}. \]
\end{lemma}

\begin{proof}
  Let $\beta \defas \pE_\rho P^\dagger Q$. WLOG $\beta = \abs{\beta}$, because $\beta = \gamma \abs{\beta}$ for some
  $\abs{\gamma} = 1$, and the special case (of $\beta = \abs{\beta}$) implies the general inequality by multiplying
  $\gamma$ to $P$.

  The rest of the proof follows \cite[Lemma~3.5]{BarakBHKSZ12}.
  Suppose $\pE_\rho P^\dagger P > 0$ and $\pE_\rho Q^\dagger Q > 0$.
  By linearity of the inequality, WLOG $\pE_\rho P^\dagger P = \pE_\rho Q^\dagger Q = 1$.
  Expand $\pE_\rho (P-Q)^\dagger(P-Q) \geq 0$, we get $2 = \pE_\rho P^\dagger P + \pE_\rho Q^\dagger Q \geq \pE
  P^\dagger Q + \pE Q^\dagger P = 2\Abs{\pE P^\dagger Q}$.

  Suppose $\pE_\rho P^\dagger P = 0$.
  Let $\alpha > 0$.
  Expand $\pE_\rho (P-\alpha Q)^\dagger(P-\alpha Q) \geq 0$, we get $\alpha^2 \pE_\rho Q^\dagger Q
  \geq \alpha \pE_\rho P^\dagger Q + \alpha \pE_\rho Q^\dagger P = 2\alpha \abs{\pE_\rho P^\dagger Q}$.
  Therefore $\abs{\pE_\rho P^\dagger Q} = 0$ by taking $\alpha \to 0$.
\end{proof}

\section{Pairwise-Uniform Avoidance for Three Qubits}
\label{sec:pairwise-uniform-state-proof}

Let $A=B=C\in \C^2$.  A state $\rho$ on $A\ot B\ot C$ is \emph{pairwise uniform} if its two-qubit marginals satisfy
\[
  \Tr_A(\rho)=\frac{I_{BC}}{4},\qquad
  \Tr_B(\rho)=\frac{I_{AC}}{4},\qquad
  \Tr_C(\rho)=\frac{I_{AB}}{4}.
\]

\begin{proposition}
  \label{prop:pairwise-uniform-avoiding-vector}
  For every unit vector $w\in A\ot B\ot C$, there exists a pairwise-uniform state $\rho$ such that $\rho w=0$.
  In particular, $\Tr(\rho\ket{w}\bra{w})=0$.
\end{proposition}

\begin{proof}
  Identify $B\ot C$ with $M_2(\C)$ by vectorization and consider the real subspace
  \[
    \mathcal F
    =
    \left\{
      \begin{pmatrix}
        a & b \\
        -\overline b & \overline a
      \end{pmatrix}
      :a,b\in\C
    \right\}.
  \]
  A direct calculation shows that the Hilbert--Schmidt inner product is real on $\mathcal F$, that
  $M_2(\C)=\mathcal F\oplus i\mathcal F$ over $\R$, and that
  \[
    XX^\dagger=X^\dagger X=\frac{\norm{X}_{\mathrm F}^2}{2}I_2
    \qquad (X\in\mathcal F).
  \]
  Thus every real orthonormal basis of $\mathcal F$ is a complex orthonormal basis of $B\ot C$ consisting of
  maximally entangled vectors.

  Fix such a basis $f_1,\ldots,f_4$ and write $w=\sum_k a_k\ot f_k$.  Regard $A$ as the real inner-product space
  $A_{\R}$ with $\langle x,y\rangle_{\R}=\Re\langle x,y\rangle$.  A real singular-value decomposition of
  $x\mapsto\sum_kx_ka_k$ gives real orthonormal bases $\{o_j\}$ of $\R^4$ and $\{u_j\}$ of $A_{\R}$ such that
  $\sum_k(o_j)_ka_k=s_ju_j$ for some $s_j\geq0$.  Setting $\beta_j=\sum_k(o_j)_kf_k$ gives
  \[
    w=\sum_{j=1}^4s_ju_j\ot\beta_j,
    \qquad
    \Tr_B\bigl(\ket{\beta_j}\bra{\beta_j}\bigr)
    =\frac{I_C}{2},
    \qquad
    \Tr_C\bigl(\ket{\beta_j}\bra{\beta_j}\bigr)
    =\frac{I_B}{2},
  \]
  where the $\beta_j$ form a complex orthonormal basis of $B\ot C$.

  Parseval for the real orthonormal bases $\{u_j\}_{j=1}^4$ and $\{iu_j\}_{j=1}^4$ of $A_{\R}$ gives
  \[
    \sum_{j=1}^4\abs{\braket{u_j|z}}^2
    =\sum_{j=1}^4\bigl(\Re\braket{u_j|z}\bigr)^2
      +\sum_{j=1}^4\bigl(\Re\braket{iu_j|z}\bigr)^2
    =2\norm{z}^2
    \qquad (z\in A).
  \]
  Hence $\sum_j\ket{u_j}\bra{u_j}=2I_A$.  The reflections
  $R_j=I_A-2\ket{u_j}\bra{u_j}$ therefore satisfy
  \[
    R_j^2=I_A,
    \qquad R_ju_j=-u_j,
    \qquad \Tr(R_j)=0,
    \qquad \sum_jR_j=0.
  \]
  Define
  \[
    H=\sum_{j=1}^4R_j\ot\ket{\beta_j}\bra{\beta_j}.
  \]
  Orthogonality of the $\beta_j$ implies that $H$ is a Hermitian involution and $Hw=-w$.  Moreover,
  \[
    \begin{aligned}
      \Tr_A(H)&=0,\\
      \Tr_B(H)&=\frac12\left(\sum_jR_j\right)\ot I_C=0,
      &\qquad
      \Tr_C(H)&=\frac12\left(\sum_jR_j\right)\ot I_B=0.
    \end{aligned}
  \]
  Consequently, $\rho=(I_{ABC}+H)/8$ is positive semidefinite, has trace one and pairwise-uniform marginals, and
  satisfies $\rho w=0$.
\end{proof}

\end{document}